\documentclass{article}

\PassOptionsToPackage{numbers, compress}{natbib}

\usepackage[preprint]{neurips_2026}

\usepackage[utf8]{inputenc}
\usepackage[T1]{fontenc}
\usepackage{hyperref}
\usepackage{url}
\usepackage{booktabs}
\usepackage{amsfonts}
\usepackage{nicefrac}
\usepackage{microtype}
\usepackage{xcolor}

\usepackage{amsmath}
\usepackage{amssymb}
\usepackage{amsthm}
\usepackage{mathtools}
\usepackage{thmtools}
\usepackage{algorithm}
\usepackage{enumitem}
\usepackage{algorithmic}



\allowdisplaybreaks[1]

\usepackage{tocloft}

\addtocontents{toc}{\protect\setcounter{tocdepth}{0}}

\declaretheorem[name=Lemma]{lemma}

\newtheorem{assumption}{Assumption}
\newtheorem{theorem}{Theorem}

\usepackage{graphicx}
\usepackage{subcaption}
\usepackage{tikz}
\DeclareRobustCommand{\circled}[1]{\tikz[baseline=(char.base)]{\node[shape=circle, draw, inner sep=1.2pt] (char) {#1};}}

\title{Independent Learning of Nash Equilibria in\\Partially Observable Markov Potential Games\\with Decoupled Dynamics}

\author{%
Philip Jordan\\
SYCAMORE, EPFL\\
\texttt{philip.jordan@epfl.ch} \\
\And
Maryam Kamgarpour\\
SYCAMORE, EPFL\\
\texttt{maryam.kamgarpour@epfl.ch}
}

\begin{document}

\maketitle

\begin{abstract}
We study Nash equilibrium learning in partially observable Markov games (POMGs), a multi-agent reinforcement learning framework in which agents cannot fully observe the underlying state. Prior work in this setting relies on centralization or information sharing, and suffers from sample and computational complexity that scales exponentially in the number of players. We focus on a subclass of POMGs with independent state transitions, where agents remain coupled through their rewards, and assume that the underlying fully observed Markov game is a Markov potential game. For this class, we present an independent learning algorithm in which players, observing only their own actions and observations and without communication, jointly converge to an approximate Nash equilibrium. Due to partial observability, optimal policies may in general depend on the full action-observation history. Under a filter stability assumption, we show that policies based on finite history windows provide sufficient approximation guarantees. This enables us to approximate the POMG by a surrogate Markov game that is near-potential, leading to quasi-polynomial sample and computational complexity for independent Nash equilibrium learning in the underlying POMG.
\end{abstract}

\section{Introduction}
\label{sec:intro}

Markov games have become the standard framework for modeling multi-agent reinforcement learning. In this framework, agents are typically assumed to have full access to the underlying state of the environment. However, in many practical settings, such as autonomous driving~\cite{lauri2022partially,bai2015intention}, multi-robot control under sensor noise~\cite{lauri2022partially,xiao2025asynchronous}, or strategic interactions in games of imperfect information~\cite{yao2020solving}, this assumption is violated due to noisy or partial state information. Partially observable Markov games (POMGs) extend the Markov game model by allowing each agent to receive only limited or noisy observations of the true state. Developing principled multi-agent learning methods for such partially observable settings remains a central challenge.

Despite the broad applicability of POMGs, our understanding of their tractability, as well as the design of efficient algorithms with provable guarantees, remains limited. In contrast, in the single-agent setting of partially observable Markov decision processes (POMDPs), it is known that optimal policies may in general depend on the entire action-observation history~\cite{puterman1994markov}, leading to computational intractability of the planning problem even for a single agent~\cite{papadimitriou1987complexity}, as well as exponential statistical complexity for learning near-optimal policies~\cite{krishnamurthy2016pac}. However, these hardness results are worst-case in nature, and recent advances paint a clearer picture of rich POMDP subclasses that admit (quasi-) efficient learning and planning~\cite{golowich2023planning,liu2022partially}. These developments motivate investigating whether similar progress can be achieved in the multi-agent setting.

\paragraph{Tractability under partial observability.} Since the main challenge in POMDPs arises from the dependence of optimal policies on the full action-observation history, a large body of work studies structural conditions under which finite-window approximations become effective. One line of work achieves this by imposing richness conditions on the observations. In particular, the $\alpha$-weakly revealing condition of~\cite{liu2022partially} yields polynomial sample complexity, disregarding computational aspects.
The related notion of $\gamma$-observability enables planning and learning algorithms with \emph{quasi-polynomial} computational and sample complexity via intricate policy cover constructions~\cite{golowich2022learning,golowich2023planning}. Furthermore, \cite{golowich2022learning} shows that under this condition, and standard complexity-theoretic assumptions, polynomial-time algorithms are ruled out.

Another family of positive algorithmic results is based on a \emph{filter stability} assumption, under which the posterior distribution over states given the action-observation history converges exponentially fast, regardless of the initial prior. This enables practical finite-window approximations and has led to provable learning algorithms closer to those used in practice, including Q-learning~\cite{kara2023convergence}, gradient methods~\cite{cayci2024finite}, and temporal difference learning~\cite{anjarlekarscalable}, that achieve quasi-polynomial computational and sample complexity~\cite{cayci2024finite,anjarlekarscalable,jordan2026model}.

These positive results in the single-agent setting raise the question of whether similar assumptions can enable provable methods for multi-agent learning under partial observability. While a large body of practical work proposes algorithms for POMGs~\cite{emery2004approximate,omidshafiei2017deep,lerer2020improving,lu2025divergence}, only few provide theoretical guarantees. From a statistical perspective, \cite{liu2022sample} extends a sample-efficient POMDP approach~\cite{liu2022partially} to POMGs under an $\alpha$-weakly revealing condition, by replacing the policy optimization oracle used in~\cite{liu2022partially} with an oracle for equilibrium computation. This leaves open the possibility of developing methods that are both statistically and computationally efficient.

Achieving computationally efficient guarantees in the multi-agent setting is challenging even under full observability. In particular, Nash equilibrium computation is intractable in general Markov games, and already PPAD-hard for normal-form games~\cite{daskalakis2009complexity}. As a result, much of the literature focuses on structured subclasses, among which Markov potential games are particularly well studied~\cite{macua2018learning,song2022can}. In this setting, several independent learning algorithms have been shown to converge provably to Nash equilibria without requiring centralized coordination and scaling to large numbers of players~\cite{leonardos2022global,ding2022independent,maheshwari2025independent}.

Partially observable games with a common reward structure, captured by decentralized POMDPs (Dec-POMDPs)~\cite{bernstein2002complexity}, form an important subclass of partially observable Markov potential games. A common approach in this setting is to assume that agents share information, effectively reducing the problem to a centralized decision process~\cite{dibangoye2016optimally,dibangoye2018learning,mao2023decentralized}. More recently, \cite{liu2023partially} establishes quasi-polynomial-time Nash equilibrium learning under $\gamma$-observability in cooperative POMGs by leveraging information sharing. However, this comes at the cost of an exponential dependence on the number of players, the so-called \emph{curse of multi-agency}. Moreover, their results suggest that observability assumptions sufficient in the single-agent setting do not directly extend to decentralized multi-agent settings, where agents may form inconsistent beliefs due to differing local observations.

\paragraph{Decoupled dynamics.}
Motivated by these challenges, we aim to design a communication-free independent learning algorithm that avoids the curse of multi-agency. To this end, we impose an additional structural assumption: while agents remain coupled through the reward function, their state dynamics are decoupled. Such models arise naturally in applications including wireless networks~\cite{altman2007constrained,altman2009dynamic}, smart energy grids~\cite{etesami2018stochastic}, autonomous driving, and multi-robot control~\cite{xiao2025asynchronous}; see Appendix~\ref{app:decoupled} for detailed illustrations. Decoupled dynamics have also been studied in the Markov game literature~\cite{zhang2023markov,etesami2024learning} and their partially observable extension defines a rich subclass of POMGs. Two concrete instances of the above applications, a dynamic transmission game over a shared channel and a cooperative navigation task, serve as the environments of our numerical experiments in Appendix~\ref{app:experiments}.

\paragraph{Contributions.}
In this work, we study POMGs with decoupled dynamics and potential structure under a filter stability assumption. Our approach builds on recent advances in approximating POMDPs by finite-state MDPs based on finite observation-action windows~\cite{kara2023convergence,golowich2022learning,cayci2024finite,anjarlekarscalable,jordan2026model}. Extending this idea to the multi-agent setting, we approximate a decoupled POMG $\mathcal{P}$ by a Markov game $\mathcal{G}^m$ defined over $m$-step history windows, termed superstate Markov game. We then establish the following:

\begin{enumerate}
\item We show that approximate Nash equilibria of the Markov game $\mathcal{G}^m$ induce approximate Nash equilibria of the original POMG with explicit error bounds (Proposition~\ref{prop:approx-ne}).
\item Based on decoupledness of transition dynamics, we establish that the superstate Markov game $\mathcal{G}^m$ admits a Markov near-potential structure~\cite{guo2025markov,guo2025potential} (Proposition~\ref{prop:alpha-potential}).
\item Building on recent independent learning methods for Markov potential games~\cite{zhang2023markov}, we propose an independent learning algorithm for the superstate Markov game $\mathcal{G}^m$ based on finite-window model estimation and soft policy iteration (see Algorithm~\ref{alg:main}).
\item We prove that the proposed algorithm converges to an $\epsilon$-approximate Nash equilibrium of the underlying POMG with quasi-polynomial computational and sample complexity, under a concentrability condition on the visitation of finite history windows. In particular, the complexity does not scale exponentially with the number of players (see Theorem~\ref{thm:main}).
\item In addition to our theoretical results, we provide simulations in two decoupled potential POMG environments, a dynamic transmission game and a cooperative gridworld, in which we observe convergence of the proposed algorithm and and increasingly accurate approximation of the equilibrium in the POMG as the window length $m$ grows (see Appendix~\ref{app:experiments}).
\end{enumerate}

The remainder of the paper is organized as follows. In Section~\ref{sec:problem}, we introduce the formal problem setting. Section~\ref{sec:mg-repr} develops the structural properties of the superstate Markov game. Building on these insights, Section~\ref{sec:algorithm} presents our independent learning algorithm, and Section~\ref{sec:results} establishes its convergence guarantees. We conclude in Section~\ref{sec:conclusion}.

\section{Problem setting}
\label{sec:problem}

Let $\mathbb{N} \coloneqq \left\{ 1,2,\dots \right\}$ and $\mathbb{N}_0 \coloneqq \mathbb{N} \cup \left\{ 0 \right\}$. For $n \in \mathbb{N}$, let $[n] \coloneqq \left\{ 1,2,\dots,n \right\}$ and $[n]_0 \coloneqq [n] \cup \left\{ 0 \right\}$. For $a,b \in \mathbb{N}$ with $a \leq b$, let $[a,b] \coloneqq \left\{ a,a+1,\dots,b-1,b \right\}$.

\paragraph{POMG.} We consider $N$-player tabular finite-horizon partially observable Markov games in which transition and observation kernels are independent across players, but rewards are coupled. Formally, a POMG within this class is defined as a tuple $\mathcal{P}=(\mathcal{N},H,\mathcal{S},\mathcal{A},\mathcal{O},r,\mathbb{P},\mathbb{O},\mu)$ with players $\mathcal{N} \coloneqq [N]$, horizon $H \in \mathbb{N}$, finite state space $\mathcal{S} \coloneqq \mathcal{S}_1 \times \dots \times \mathcal{S}_N$, finite action space $\mathcal{A} \coloneqq \mathcal{A}_1 \times \dots \times \mathcal{A}_N$, and finite observation space $\mathcal{O} \coloneqq \mathcal{O}_1 \times \dots \times \mathcal{O}_N$. Player $i$'s reward function at step $h \in [H]$ is defined as $r_{i,h}:\mathcal{S} \times \mathcal{A} \to [0,1]$. State transition probabilities are given by $\mathbb{P}=\left( \mathbb{P}_{i,h} \right)_{i \in \mathcal{N},h \in [H]}$ with each player's transition $\mathbb{P}_{i,h}(\cdot \mid s_{i,h},a_{i,h}) \in \Delta(\mathcal{S}_i)$ for $s_{i,h} \in \mathcal{S}_i$, $a_i \in \mathcal{A}_i$, and observation probabilities are given by $\mathbb{O}=\left( \mathbb{O}_{i,h} \right)_{i \in \mathcal{N}, h \in [H]}$ with $\mathbb{O}_{i,h}(\cdot \mid s_{i,h}) \in \Delta(\mathcal{O}_i)$ for each $i \in \mathcal{N}$. Observe that rewards are coupled through the dependence on joint states and actions, whereas state transitions and observations factor across players, as detailed below.

At step $h \in [H]$, when in state $s_h \in \mathcal{S}$, joint observation $o_h \in \mathcal{O}$ is emitted with probability $\mathbb{O}_h(o_h \mid s_h) = \prod_{i \in \mathcal{N}} \mathbb{O}_{i,h}(o_{i,h} \mid s_{i,h})$, reflecting the decoupled observation structure, and each player $i \in \mathcal{N}$ observes its component $o_{i,h}$. Then, player $i$ chooses an action $a_{i,h} \in \mathcal{A}_i$ and receives reward $r_{i,h}(s_h,a_h)$. Next, the unobserved state transitions to $s_{h+1} \in \mathcal{S}$ with probability $\mathbb{P}_h(s_{h+1} \mid s_h,a_h) = \prod_{i \in \mathcal{N}} \mathbb{P}_{i,h}(s_{i,h+1} \mid s_{i,h},a_{i,h})$, reflecting the decoupled transition structure. The initial state is drawn from the product distribution $\mu \in \Delta(\mathcal{S})$, i.e., $\mu(s_1) = \prod_{i \in \mathcal{N}} \mu_i(s_{1,i})$.

\paragraph{Histories and policies.} For each player $i \in \mathcal{N}$ and $h \in \mathbb{N}$, let $\mathcal{H}_i^h \coloneqq (\mathcal{O}_i \times \mathcal{A}_i)^{h-1} \times \mathcal{O}_i$ be the set of $h$-step action-observation histories, and let $\mathcal{H}_i^0$ denote the empty history. Let $\mathcal{H}_i^{\leq h} \coloneqq \bigcup_{h^{\prime} \in [h]_0} \mathcal{H}_i^{h^{\prime}}$. Following the standard POMDP convention (see, e.g.,~\cite{golowich2022learning}), histories are observation-first and end with the current observation: let $\tau_i \in \mathcal{H}_i^h$ be a history written as $\tau_i = (o_{i,1},a_{i,1},\dots,a_{i,h-1},o_{i,h})$. For $k,l \in \mathbb{N}$ with $k \leq l$, we define the sub-history $\tau_{i,k:l} = (o_{i,k},a_{i,k},\dots,a_{i,l-1},o_{i,l})$. In particular, $\tau_{i,1:h} = \tau_i$. We denote by $|\tau_i| = h$ the length of $\tau_i$. For each player $i \in \mathcal{N}$, we define the class of history-dependent policies as
\begin{align*}
\Pi_i^H \coloneqq \left\{ \pi_i = (\pi_{i,h})_{h \in [H]} \mid \pi_{i,h}:\mathcal{H}_i^{h} \to \Delta(\mathcal{A}_i) \text{ for all } h \in [H] \right\}.
\end{align*}
We consider localized product policies $\pi = (\pi_1,\dots,\pi_N) \in \Pi_1^H \times \dots \times \Pi_N^H$, meaning that at step $h$, given local histories $\tau=(\tau_1,\dots,\tau_N) \in \mathcal{H}^{h}$, joint action $a = (a_1,\dots,a_N) \in \mathcal{A}$ is chosen by $\pi \in \Pi^H$ with probability $\pi_h(a \mid \tau) = \prod_{i \in \mathcal{N}} \pi_{i,h}(a_i \mid \tau_i)$.

Moreover, for any joint policy $\pi \in \Pi^H$, we define for each $i \in \mathcal{N}$ the value function
\begin{align}
\label{eqn:value-def}
V_i(\pi) \coloneqq \mathbb{E}_{\pi,\, s_1 \sim \mu} \left[ \sum_{h=1}^H r_{i,h}(s_h,a_h) \right],
\end{align}
where the expectation is taken over the distribution of trajectories induced by the POMG $\mathcal{P}$ under the joint policy $\pi$, including the randomness from state transitions and observation emissions.

\paragraph{Nash equilibria.}
Let $\epsilon > 0$. A joint policy $\pi \in \Pi^H$ is called an $\epsilon$-approximate Nash equilibrium if for all players $i \in \mathcal{N}$ and all $\pi^{\prime}_i \in \Pi_i^H$, it holds that $V_i(\pi) \geq V_i(\pi^{\prime}_i, \pi_{-i}) - \epsilon$. When $\epsilon=0$, we call $\pi$ a Nash equilibrium.

Our objective in this paper is to provide an algorithm for learning $\epsilon$-approximate Nash equilibria.

\paragraph{Potential structure.}
The hardness of Nash equilibrium computation in the general-sum regime~\cite{daskalakis2009complexity} motivates us to focus on a tractable subclass such as potential games.
For player $i \in \mathcal{N}$, let $\Pi_i \coloneqq \{ \pi_i = (\pi_{i,h})_{h=1}^H \mid \pi_{i,h} : \mathcal{S} \to \Delta(\mathcal{A}_i) \}$ be the set of Markov (state-feedback) policies, and let $\Pi \coloneqq \Pi_1 \times \dots \times \Pi_N$ denote the induced Markov product policies. Overloading the notation of~\eqref{eqn:value-def}, we write, for $\pi \in \Pi$,
\begin{align}
\label{eqn:value-def-2}
V_i(\pi) \coloneqq \mathbb{E}_{\pi,\, s_1 \sim \mu}\Big[\textstyle\sum_{h=1}^H r_{i,h}(s_h, a_h)\Big],
\end{align}
where the expectation is over joint state-action trajectories of the underlying fully observed Markov game with rewards $r_{i,h}$, transitions $\mathbb{P}$, and initial distribution $\mu$, when following policy the joint $\pi$.
\begin{assumption}
\label{ass:potential}
There exist functions $\phi_h : \mathcal{S} \times \mathcal{A} \to \mathbb{R}$, $h \in [H]$, such that for all $i \in \mathcal{N}$, $s_{-i} \in \mathcal{S}_{-i}$, $a_{-i} \in \mathcal{A}_{-i}$, $s_i,s_i^{\prime} \in \mathcal{S}_i$, and $a_i,a_i^{\prime} \in \mathcal{A}_i$,
\begin{align*}
&r_{i,h}(s_i^{\prime},s_{-i},a_i^{\prime},a_{-i})-r_{i,h}(s_i,s_{-i},a_i,a_{-i}) \\
&\qquad\qquad = \phi_h(s_i^{\prime},s_{-i},a_i^{\prime},a_{-i})-\phi_h(s_i,s_{-i},a_i,a_{-i}).
\end{align*}
\end{assumption}
That is, at every step the stage game admits a potential $\phi_h$. This stagewise potential condition is common in the decoupled Markov game setting (Definition~4 of~\cite{zhang2023markov}) and is typically easier to verify than a global potential condition on value functions. Moreover, under decoupled transitions it implies the global Markov potential game condition, namely the existence of $\Phi:\Pi \to \mathbb{R}$ with $V_i(\pi)-V_i(\pi_i^{\prime},\pi_{-i}) = \Phi(\pi) - \Phi(\pi_i^{\prime},\pi_{-i})$ for all $i \in \mathcal{N}$, $\pi \in \Pi$, and $\pi_i^{\prime} \in \Pi_i$, which underlies a large body of work, see~\cite{leonardos2022global,zhang2022global,zhang2023markov} and others. Assumption~\ref{ass:potential} is satisfied in the fully cooperative case $r_1=\dots=r_N$, as well as in practical settings such as demand-response markets~\cite{narasimha2022multi}, and other mixed cooperative/competitive scenarios for which \cite{narasimha2022multi} provides several sufficient conditions.

Before turning to the main objective of Nash equilibrium learning, the next section first establishes structural properties of the considered games which will then guide our algorithm design in Section~\ref{sec:algorithm}.

\section{Approximating POMGs via Markov near-potential games}
\label{sec:mg-repr}

In this section, we show that any decoupled POMG can be approximated by a Markov game, which we term the \emph{superstate Markov game}, obtained via a finite-window approximation based on truncated histories. This extends related approaches from the single-agent setting~\cite{kara2023convergence,golowich2022learning,cayci2024finite} to the multi-agent case. We establish that the resulting superstate Markov game approximately preserves Nash equilibria of the original POMG (Proposition~\ref{prop:approx-ne}) and admits a near-potential structure (Proposition~\ref{prop:alpha-potential}).

\paragraph{Superstate Markov game.} Let $m \geq 1$. The superstate Markov game $\mathcal{G}^m$ has state space $\mathcal{H}^{\leq m}$ and action space $\mathcal{A}$. Concretely, at step $h$, the state is given by $w = (w_i)_{i \in \mathcal{N}}$, where we use the notation $w_i=(o_{i,h-m+1}, a_{i,h-m+1}, \dots, a_{i,h-1}, o_{i,h}) \in \mathcal{H}_i^{\min(h,m)}$ for the most recent $m$-step local history of player $i$.
To define the transition dynamics of the superstate Markov game, we introduce a belief over latent states induced by truncated histories. In general, a truncated history does not determine the posterior distribution of the underlying state. We therefore define, for each player~$i \in \mathcal{N}$ and step~$h \in [H]$, the belief $b^m_{i,h}(s_i \mid w_i)$: the probability that player $i$ is in state $s_i$ at step $h$ after observing the window $w_i$, starting from the initial distribution\footnote{We use $\mu_i$ as a prior on the state at the start of the window; any prior over $\mathcal{S}_i$ could be used instead.} $\mu_i$. Formally, writing $s_{i,h} = s_i$, this is obtained by marginalizing over all local state trajectories $s_{i,h-m+1},\dots,s_{i,h-1}$ consistent with $w_i$,
\begin{align*}
b^{m}_{i,h}(s_i \mid w_i) &=
\frac{1}{Z_i(w_i)}
\sum_{s_{i,h-m+1},\dots,s_{i,h-1} \in \mathcal{S}_i}
\mu_i(s_{i,h-m+1}) \, \mathbb{O}_{i,\, h-m+1}(o_{i,h-m+1} \mid s_{i,h-m+1}) \\
&\qquad \cdot \prod_{k=h-m+1}^{h-1} \mathbb{P}_{i,\, k}(s_{i,k+1} \mid s_{i,k}, a_{i,k}) \;
\mathbb{O}_{i,\, k+1}(o_{i,k+1} \mid s_{i,k+1}),
\end{align*}
where the normalization factor $Z_i(w_i)$ is the sum over $s_i \in \mathcal{S}_i$ of the unnormalized probabilities of reaching $s_i$. Intuitively, this expression averages over all latent state sequences that could have generated the observed window $w_i$, each weighted by its likelihood under the model.
For $h \in [H-1]$, $w,w^{\prime} \in \mathcal{H}^{\leq m}$, $a \in \mathcal{A}$, and $o \in \mathcal{O}$ with $|w^{\prime}| = n$ and $w' = (w \circ (a,o))_{n-m+2:n+1}$,
\begin{align}
\label{eqn:Pm-def}
\mathbb{P}_h^m(w' \mid w,a)
\;\coloneqq\;
\sum_{s,s^{\prime} \in \mathcal{S}}
\Big( \prod_{i \in \mathcal{N}} b^m_{i,h}(s_i \mid w_i) \Big)
\mathbb{P}_h(s^{\prime} \mid s,a) \,
\mathbb{O}_{h+1}(o \mid s^{\prime}).
\end{align}
Otherwise, if $w^{\prime}$ cannot be obtained by concatenating $w$ with any action-observation pair, we set $\mathbb{P}_h^m(w' \mid w,a) = 0$.
Similarly, we define rewards $r_{i,h}^m(w,a) \coloneqq \sum_{s \in \mathcal{S}} r_{i,h}(s,a) \prod_{j \in \mathcal{N}} b^m_{j,h}(s_j \mid w_j)$. Moreover,
for policy $\pi \in \Pi^m$, its value in $\mathcal{G}^m$ is defined as
\begin{align*}
V^m_i(\pi) \coloneqq \mathbb{E}_{\pi,\, s_1 \sim \mu}\left[ \sum_{h=1}^H r^m_{i,h}(w_h,a_h) \right]
\end{align*}
where the expectation is taken over the randomness of the policy and the superstate transitions $\mathbb{P}^m$.

By construction, $\mathcal{G}^m$ is a Markov game with a finite state space. Moreover, due to the decoupled structure of the underlying dynamics and observations, its transition kernel factorizes across players. Importantly, $\mathcal{G}^m$ is a conceptual object used for analysis: it is defined via a finite-window truncation of histories, which does not correspond to the true evolution of observations in the POMG. In particular, trajectories are generated by the underlying POMG, where observations depend on the full history, and therefore do not follow the superstate transition kernel $\mathbb{P}^m$.

We note that for $m = H$, we obtain a Markov game with state space corresponding to the full action-observation histories. This model is equivalent to the POMG in the sense that any joint policy over histories induces the same trajectory distribution and value in both models.

\subsection{Finite-window policies and filter stability}

Having defined the superstate Markov game $\mathcal{G}^m$, we aim to establish that its Nash equilibria over finite-window policies correspond to approximate Nash equilibria of the original POMG. Such a guarantee would justify using $\mathcal{G}^m$ as a surrogate game in the analysis of learning algorithms.

In order to argue about finite-window Nash equilibria, we first introduce the respective policy class.
\paragraph{Finite-window policies.} For any $m \in [H]$, and each player $i \in \mathcal{N}$, let
\begin{align*}
\Pi_i^m \coloneqq \big\{ \pi_i = (\pi_{i,h})_{h \in [H]} \mid \pi_{i,h}:\mathcal{H}_i^{\min(h,m)} \to \Delta(\mathcal{A}_i) \text{ for all } h \in [H] \big\}.
\end{align*}
We further define joint finite-window policies given by product policies $\pi = (\pi_1,\dots,\pi_N) \in \Pi_1^m \times \dots \times \Pi_N^m$, meaning that at step $h$, given local histories $\tau=(\tau_1,\dots,\tau_N) \in \mathcal{H}^{h}$, $a \in \mathcal{A}$ is chosen with probability $\pi^{\prime}(a \mid \tau_{h-m+1 \,:\, h}) = \prod_{i \in \mathcal{N}} \pi^{\prime}_i(a_i \mid \tau_{i, h-m+1 \,:\, h})$. We note that $\Pi^m$ can be seen as a subclass of $\Pi^H$ by associating to each $\pi \in \Pi^m$ its extension to $\Pi^H$ that depends only on the most recent $m$ steps of history. Thus, for $i \in \mathcal{N}$ and $\pi \in \Pi^m$, the value $V_i(\pi)$ is defined as in~(\ref{eqn:value-def}).

Without further assumptions, relevant information about the state may reside arbitrarily far in the past, and it is not clear whether finite-window approximation can yield meaningful guarantees. We therefore introduce a standard condition that ensures sufficient decay of past information.

\paragraph{Filter stability.} Filter stability formalizes the idea that the influence of initial beliefs decays exponentially over time, and has been studied in prior work on hidden Markov models~\cite{van2008hidden} and POMDPs~\cite{kara2023convergence,anjarlekarscalable}.

For $h \in [H]$ and joint history $\tau = (\tau_1,\dots,\tau_N) \in \mathcal{H}^h$, let $b_h(\cdot \mid \tau) \in \Delta(\mathcal{S})$ denote the posterior distribution over joint states $\mathcal{S}$. By decoupledness of the dynamics and observations, it factorizes into the playerwise beliefs, $b_h(s \mid \tau) = \prod_{i \in \mathcal{N}} b_{i,h}(s_i \mid \tau_i)$. For $\nu,\nu^{\prime} \in \Delta(\mathcal{S})$, define the total variation distance $\|\nu-\nu^{\prime}\|_{TV} \coloneqq \frac{1}{2} \sum_{s \in\mathcal{S}} |\nu(s)-\nu^{\prime}(s)|$.
\begin{assumption}
\label{ass:filter-stab}
There exists $\rho > 0$ such that for all $h \in [H]$, $\tau,\tau^{\prime} \in \mathcal{H}^h$, $a \in \mathcal{A}$, and $o \in \mathcal{O}$,
\begin{align*}
\left\lVert b_h(\cdot \mid \tau \circ (a,o)) - b_h(\cdot \mid \tau^{\prime} \circ (a,o)) \right\rVert_{TV}
\leq (1-\rho) \left\lVert b_h(\cdot \mid \tau)-b_h(\cdot \mid \tau^{\prime}) \right\rVert_{TV}
\end{align*}
where $\circ$ denotes concatenation.
\end{assumption}
Filter stability requires (a) state transitions to be sufficiently mixing, and (b) observations to be sufficiently noisy. We refer to~\cite{kara2022near} for a sufficient condition that formalizes these requirements in terms of Dobrushin coefficients of the transition and observation kernel. Although stated for the joint belief, Assumption~\ref{ass:filter-stab} hides no exponential dependence on the number of players; see the remark following Lemma~\ref{lem:beliefs-tv} in Appendix~\ref{app:proof-mg-repr}.

\subsection{Finite-window Nash equilibrium approximation and near-potential structure}
Based on the filter stability assumption above, we aim to show that Nash equilibria of the superstate Markov game over the finite-window policy class $\Pi^m$ correspond to approximate Nash equilibria of the POMG over the full-history class $\Pi^H$. This approximation result will allow us to analyze equilibria in the tractable superstate setting while retaining guarantees for the original POMG. In addition, we establish that the superstate Markov game admits a near-potential structure, which will enable us to prove convergence of our learning algorithm.

For the single-agent case, it is known that under filter stability, the optimal value of a superstate MDP approximates the optimal value of the corresponding POMDP, see Theorem~2 of~\cite{anjarlekarscalable}. We extend this guarantee to the multi-agent setting by establishing value approximation results for joint policies, as well as for best-response values. Notably, this result, as well as Proposition~\ref{prop:approx-ne} below, hold for general POMGs and do not rely on the decoupled dynamics assumption.
\begin{restatable}{lemma}{valueapproxlemma}
\label{lem:value-approx}
Let $\pi \in \Pi^m$ and $\pi^{\prime} \in \Pi^H$ such that for all $h \geq m$ and all $w \in \mathcal{H}^m$, $\tau \in \mathcal{H}^h$ with $w=\tau_{h-m+1:h}$, it holds that $\pi_h(\cdot \mid w)=\pi^{\prime}_h(\cdot \mid \tau)$. Let
\begin{align}
\label{eqn:eps-def}
\epsilon^m_\rho \coloneqq 4H^2(1-\rho)^{m-1}.
\end{align}
Under Assumption~\ref{ass:filter-stab}, for any $i \in \mathcal{N}$, we have
\begin{align*}
\left\vert V^m_i(\pi)-V_i(\pi^{\prime}) \right\vert \leq \epsilon^m_\rho.
\end{align*}
Moreover, for $\pi$ and $\pi^{\prime}$ as above, it holds that
\begin{align*}
\left\vert \max_{\hat{\pi}_i \in \Pi^m_i}V^m_i(\hat{\pi}_i,\pi_{-i}) - \max_{\hat{\pi}_i^{\prime} \in \Pi^H_i}V_i(\hat{\pi}_i^{\prime},\pi_{-i}) \right\vert \leq \epsilon^m_\rho.
\end{align*}
\end{restatable}

We are then able to connect equilibria of the superstate Markov game to those of the original POMG.
\begin{restatable}{proposition}{approxneprop}
\label{prop:approx-ne}
Let $\epsilon > 0$, and let $\pi \in \Pi^m$ be an $\epsilon$-approximate Nash equilibrium of the superstate Markov game $\mathcal{G}^m$. Under Assumption~\ref{ass:filter-stab}, $\pi$ is an $(\epsilon + 2\epsilon^m_\rho)$-approximate Nash equilibrium of $\mathcal{P}$.
\end{restatable}
The proofs of Lemma~\ref{lem:value-approx} and Proposition~\ref{prop:approx-ne} are provided in Appendix~\ref{app:proof-mg-repr}.

The following result shows that $\mathcal{G}^m$ inherits a \emph{near-potential} structure~\cite{guo2025markov}, suggesting that efficient learning in the superstate Markov game is possible.
\begin{restatable}{proposition}{alphapotentialprop}
\label{prop:alpha-potential}
Under Assumptions~\ref{ass:potential} and~\ref{ass:filter-stab}, there exists a function $\Psi: \Pi^m \to \mathbb{R}$ such that for all $\pi \in \Pi^m$, $i \in \mathcal{N}$, and $\pi_i^{\prime} \in \Pi^m_i$, and $\epsilon_\rho^m$ as in~(\ref{eqn:eps-def}), we have
\begin{align*}
\left\vert \left( V^m_i({\pi})-V^m_i(\pi_i^{\prime},\pi_{-i}) \right) - \left( \Psi(\pi) - \Psi(\pi_i^{\prime},\pi_{-i}) \right) \right\vert \leq 2 \epsilon^m_\rho.
\end{align*}
\end{restatable}
The proof is deferred to Appendix~\ref{app:proof-mg-repr}. At a high level, the decoupled structure allows state visitation distributions to factor across players, enabling the construction of a potential function for the POMG, which then transfers approximately to the finite-window setting.

Building upon these structural properties, we now turn to the question of learning in POMGs.

\section{Independent Nash equilibrium learning algorithm}
\label{sec:algorithm}

In this section, we design a learning algorithm for decoupled POMGs with potential structure. We begin by stating the learning objective and protocol, then we outline our approach. An implementation of the resulting algorithm, along with simulation results in two decoupled potential POMG environments, is provided in Appendix~\ref{app:experiments}.

\paragraph{Learning Objective.} Given $\epsilon > 0$, our objective is to efficiently learn a finite-window policy $\pi \in \Pi^m$ such that $\pi$ is an $\epsilon$-approximate Nash equilibrium of the POMG, that is, for all $i \in \mathcal{N}$ and all deviations with full history dependence $\pi^{\prime}_i \in \Pi_i^H$, it holds that $V_i(\pi) \geq V_i(\pi^{\prime}_i, \pi_{-i}) - \epsilon$.

\paragraph{Independent learning protocol.} We consider an independent learning setting~\cite{zhang2021multi,daskalakis2020independent,ding2022independent} in which players interact with the game over a fixed number of episodes. Within each episode, every player follows a fixed policy $\pi_i \in \Pi^m$, observes only its own sequence of actions, observations, and rewards, and then updates its policy individually. Players do not observe other agents' actions, observations, rewards, or policies. Moreover, no information is shared among players or with any central entity. We do assume, however, that all players agree in advance to follow the same algorithm, which places the resulting dynamics in the self-play regime~\cite{bai2020provable}.

Our algorithm (Algorithm~\ref{alg:main}) closely follows the method of~\cite{zhang2023markov} for fully observable Markov games, applied to the superstate game. The main challenges arise in the analysis, which we outline in the proof overviews in Section~\ref{sec:results}.

\begin{algorithm}
\caption{Learning Finite-Window Policy in POMG (at player $i$)}
\label{alg:main}
\begin{algorithmic}[1]
\STATE \textbf{Input:} iteration number $K$, episode number $T$, window length~$m$, stepsizes $\eta^{(k)}$ for $k \in [K]$, exploration probability $\zeta \in (0,1]$.
\vspace{.3em}
\STATE Initialize $\pi_{i,h}^{(0)} = 1 / |\mathcal{A}_i|$ for all $h \in [H]$.
\FOR{$k = 0,1,\dots,K - 1$}
\STATE Collect trajectory $\tau_i=((a^{(t)}_{i,h},o^{(t)}_{i,h},r_{i,h}^{(t)})_{h=1}^{H})_{t=1}^{T}$ by following $\pi_i^{\zeta,(k)} \coloneqq \zeta \mathcal{U}(\mathcal{A}_i)+(1-\zeta)\pi_i^{(k)}$, i.e., $a^{(t)}_{i,h} \sim \mathcal{U}(\mathcal{A}_i)$ w.p.\ $\zeta$ and  $a^{(t)}_{i,h} \sim \pi_{i,h}^{(k)}(\cdot \mid w_{i,h}^{(t)})$ otherwise, for all $h \in [H]$ and $t \in [T]$. \label{ln:eps-greedy}
\STATE Estimate $\hat{\mathbb{P}}_{i,h}^{m,\pi_i^{(k)}}$ and $\hat{r}^{m,\pi^{(k)}_{-i}}_i$ based on empirical frequencies (see Appendix~\ref{app:alg-details} for details).
\STATE Let $\hat{Q}^{(k)}_{i,H+1}(w_i,a_i) \coloneqq 0$ for all $w_i \in \mathcal{H}^{\leq m}_i, a_i \in \mathcal{A}_i$.
\FOR{$h = H,\dots,1$}
\FOR{$w_i \in \mathcal{H}^{\leq m}_i, a_i \in \mathcal{A}_i$}
\STATE $\hat{Q}^{(k)}_{i,h}(w_i,a_i) \hspace{-1pt}\coloneqq\hspace{-1pt} \hat{r}_{i,h}^{m,\pi_{-i}^{(k)}}(w_i,a_i) +\hspace{-.35em} \displaystyle\sum_{w_i^{\prime},a_i^{\prime}} \hspace{-2pt} \hat{\mathbb{P}}_{i,h}^{m,\pi_i^{(k)}}(w_i^{\prime} \mid w_i,a_i) \pi^{(k)}_{i,h+1}(a_i^{\prime} \mid w_i^{\prime}) \hat{Q}^{(k)}_{i,h+1}(w_i^{\prime},a^{\prime}_i)$. \label{ln:Q-update}
\STATE $\pi_{i,h}^{(k+1)}(a_i \mid w_i) \hspace{-1pt}\coloneqq\hspace{-1pt} (1-\eta^{(k)}) \pi_{i,h}^{(k)}(a_i \mid w_i) + \eta^{(k)} \mathbf{1}\left\{ a_i \in \arg\max_{a_i^{\prime}} \hat{Q}^{(k)}_{i,h}(w_i,a_i^{\prime}) \right\}$. \label{ln:soft-policy-iter}
\ENDFOR
\ENDFOR
\ENDFOR
\vspace{.5em}
\end{algorithmic}
\end{algorithm}

\subsection{Independent soft policy iteration}

A key challenge in independent learning in Markov games, even under decoupled transitions, is that each player's value function depends on the joint policy of all players through the coupled rewards. However, when fixing any~$\pi_{-i}$ in the superstate Markov game $\mathcal{G}^m$, player $i$ faces an MDP $\mathcal{M}^{\pi_{-i}}$ with state space $\mathcal{S}_i$, action space $\mathcal{A}_i$, transition kernel $\mathbb{P}_i^m$ as defined in (\ref{eqn:Pm-def}), and reward
\begin{align*}
r_{i,h}^{m,\pi_{-i}}(w_i,a_i) &\coloneqq \mathbb{E}_{w_{-i} \sim d^{m,\pi}_{-i,h}, a_{-i} \sim \pi_{-i}(\cdot \mid w_{-i})} \left[ r_{i,h}^m(w_i,w_{-i},a_i,a_{-i}) \right],
\end{align*}
where $d^{m,\pi}_{i,h}(w_{i,h}) \coloneqq P_\pi(w_{i,h}=w_i)$, and $d^{m,\pi}_{-i,h}(w_{-i,h}) \coloneqq \prod_{j \in \mathcal{N} \setminus \{i\}} d^{m,\pi}_{j,h}(w_j)$. Note that here $P_\pi$ is taken over the randomness of policy $\pi$ and transition kernel $\mathbb{P}^m$. We further denote the Q-function of $\mathcal{M}^{\pi_{-i}}$ by $\overline{Q}^{m,\pi}_{i,h}: \mathcal{H}_i^{\leq m} \times \mathcal{A}_i \to \mathbb{R}$.

Observe that if a policy $\pi^{\star} \in \Pi^m$ is such that for all $i \in \mathcal{N}$, $\pi^{\star}_i$ achieves the optimal value in $\mathcal{M}^{\pi^{\star}_{-i}}$, then $\pi^{\star}$ is a Nash equilibrium of $\mathcal{G}^m$. This motivates having each player perform policy iteration with respect to $\overline{Q}^{m,\pi}_{i,h}$ in the hope of converging to a joint policy that satisfies this playerwise optimality in~$\mathcal{M}^{\pi_{-i}}$. However, in a game setting, such simultaneous updates may lead to cyclic patterns that fail to converge~\cite{zinkevich2005cyclic}. In contrast, the smoothed version with appropriately chosen stepsize $0 < \eta^{(k)} < 1$, as implemented in Line~\ref{ln:soft-policy-iter} of Algorithm~\ref{alg:main}, has been shown to converge to an approximate Nash equilibrium in Markov potential games~\cite{zhang2023markov}.

It remains to describe how to obtain $\overline{Q}^{m,\pi^{(k)}}_{i,h}$. As the Q-function of an MDP, namely $\mathcal{M}^{\pi^{(k)}_{-i}}$, $\overline{Q}^{m,\pi^{(k)}}_{i,h}$ satisfies the Bellman expectation equation, which allows us to compute an estimate $\smash{\hat{Q}^{(k)}_{i,h}}$ via backward iteration (see Line~\ref{ln:Q-update} of Algorithm~\ref{alg:main}). This computation relies on having estimates $\smash{\hat{\mathbb{P}}_{i,h}^{m,\pi_i^{(k)}}}$ and $\hat{r}_{i,h}^{m,\pi^{(k)}_{-i}}$ of the transition kernel of $\smash{\mathcal{M}^{\pi_{-i}^{(k)}}}$. We next describe how to obtain these estimates for the superstate game from interaction with the POMG at each iteration.

\subsection{Sampling \& model estimation}

As we do not assume access to transition kernels or a generative model, players sample action-observation sequences over $T$ episodes by simultaneously interacting with the POMG and receiving their respective observations and rewards. At iteration~$k$ and step $h$, in order to ensure exploration of actions, player $i \in \mathcal{N}$ takes an action uniformly at random with probability $\zeta$, and otherwise follows its current policy $\smash{\pi_{i,h}^{(k)}}$. Based on the observed sequences $\smash{\tau_i = \{ \{ a_{i,h}^{(t)}, o_{i,h}^{(t)}, r_{i,h}^{(t)} \}_{h=1}^H \}_{t=1}^T}$, we apply standard model estimation techniques inspired by single-agent methods~\cite{azar2012sample,kearns2002near} as detailed below.

\paragraph{Transition estimation.}
For each player $i \in \mathcal{N}$, we estimate the superstate transition probabilities $\mathbb{P}_{i,h}^m(\cdot \mid w_i,a_i)$ by empirical frequencies of $m$-step window transitions. Concretely, for each pair $(w_i,a_i) \in \mathcal{H}_i^{\leq m} \times \mathcal{A}_i$, we collect all time steps at which the $m$-step history window equals $w_i$ and action $a_i$ is played by policy $\pi_i^{(k)}$, and record the resulting next window $w_i' \in \mathcal{H}_i^{\leq m}$. The estimate $\hat{\mathbb{P}}_{i,h}^{m,\pi_i^{(k)}}(w_i' \mid w_i,a_i)$ is given by the relative frequency of observing $w_i'$ among these occurrences. If a pair $(w_i,a_i)$ is not observed, the estimate is set to~$0$.

\paragraph{Reward estimation.}
Similarly, we estimate the reward function $r_{i,h}^{m,\pi^{(k)}_{-i}}$ of $\smash{\mathcal{M}^{\pi_{-i}^{(k)}}}$ by empirical averaging. For each $\smash{(w_i,a_i) \in \mathcal{H}_i^{\leq m} \times \mathcal{A}_i}$, we average the rewards observed at time steps where the history window equals $w_i$ and action $a_i$ is taken. If $(w_i,a_i)$ is not observed, the estimate is set to~$0$.

For formal definitions of the estimates $\hat{\mathbb{P}}_i^m$ and $\hat{r}_{i,h}^{m,\pi^{(k)}_{-i}}$, we refer to Appendix~\ref{app:alg-details}. A key challenge towards establishing guarantees for Algorithm~\ref{alg:main} lies in deriving confidence bounds for this sampling procedure. We provide further insight on this in the proof overview of Lemma~\ref{lem:est-prob} in the next section.

\section{Convergence to Nash equilibrium}
\label{sec:results}

In this section, we present our main result, Theorem~\ref{thm:main}, which provides the convergence guarantee for Algorithm~\ref{alg:main} towards an approximate Nash equilibrium.

Our analysis relies on establishing confidence bounds for the estimated superstate model~$\hat{\mathbb{P}}_i^m$ and~$\hat{r}_{i,h}^{m,\pi_{-i}}$ (see Lemma~\ref{lem:est-prob} below). This in turn requires that all $m$-step history windows are well explored. A sufficient condition is that, at each superstate, all actions are selected and all observations are realized with probabilities bounded away from zero. For actions, this is ensured by the $\zeta$-greedy exploration strategy in Line~\ref{ln:eps-greedy} of Algorithm~\ref{alg:main}. For observations, we impose the following assumption.
\begin{assumption}
\label{ass:uni-obs}
There exists $\beta > 0$ such that for all $i \in \mathcal{N}$, $s_i \in \mathcal{S}_i$ and $o_i \in \mathcal{O}_i$, $\mathbb{O}_i(o_i \mid s_i) \geq \beta$.
\end{assumption}
We note that Assumption~\ref{ass:uni-obs} can be enforced, for example, by having each player inject a small amount of uniform noise into its observations, ensuring that every observation occurs with probability bounded away from zero. Next, based on this assumption, we state the estimation confidence bounds.

\begin{restatable}{lemma}{estproblemma}
\label{lem:est-prob}
Let Assumptions~\ref{ass:filter-stab} and~\ref{ass:uni-obs} hold, and let $\pi \in \Pi^m$. Suppose each player $i \in \mathcal{N}$ collects a trajectory according to Line~\ref{ln:eps-greedy} in Algorithm~\ref{alg:main} with exploration probability $\zeta \in (0,1]$, and estimates $\hat{\mathbb{P}}^{m,\pi_i}_i$ and $\hat{r}^{m,\pi_{-i}}_i$ according to~(\ref{eq:sample-p}) and~(\ref{eq:sample-r}), respectively. Let $i \in \mathcal{N}$, $h \in [H]$, $a_i \in \mathcal{A}_i$, and $w_i,w_i^{\prime} \in \mathcal{H}_i^{\leq m}$. Then for any $0 < \epsilon \leq 1$,
\begin{align*}
P \left( \vert \mathbb{P}^m_{i,h}(w_i^{\prime} \mid w_i,a_i)-\hat{\mathbb{P}}_{i,h}^{m,\pi_i}(w_i^{\prime} \mid w_i,a_i) \vert \geq \epsilon + (1-\rho)^{m-1} \right)
&\leq 4 \exp \left( -\frac{\epsilon^{2} \zeta^{2m} \beta^{2m} T}{32 |\mathcal{A}_i|^{2m}} \right), \\[2pt]
P \left( \vert r_{i,h}^{m,\pi_{-i}}(w_i,a_i)-\hat{r}_{i,h}^{m,\pi_{-i}}(w_i,a_i) \vert \geq \epsilon + 4 N H \zeta + N^2H(1-\rho)^{m-1} \right)
&\leq 4 \exp \left( -\frac{\epsilon^{2} \zeta^{2m} \beta^{2m} T}{32 |\mathcal{A}_i|^{2m}} \right).
\end{align*}
\end{restatable}
\paragraph{Proof overview.}
The main challenge we face is a mismatch between superstate Markov game and the underlying POMG process: action-observation sequences are collected from the POMG kernel with full history dependence $\mathbb{P}^H$, rather than the superstate Markov game kernel $\mathbb{P}^m$. While this issue arises already in POMDPs~\cite{cayci2024finite,anjarlekarscalable,jordan2026model}, in the game setting the mismatch additionally enters through the other agents: player $i$'s marginalized reward $r^{m,\pi_{-i}}_{i,h}$ averages over the other agents' superstate distributions, which are induced by $\mathbb{P}^m$, whereas in the collected data they follow the POMG kernel $\mathbb{P}^H$. Thus, the approximation errors in all other agents' dynamics propagate into player $i$'s reward estimate. We handle this by decomposing the estimation error of $\hat{\mathbb{P}}_{i,h}^{m,\pi_i}$ and $\hat{r}_{i,h}^{m,\pi_{-i}}$ into a statistical term arising from randomness of the policy and environment, and a bias term due to the finite-window approximation (Lemma~\ref{lem:est-prob}). The first term is controlled by a concentration argument across the independently sampled episodes, while the bias term is bounded under Assumption~\ref{ass:filter-stab} by analyzing how mismatch propagates through other agents' states (Lemma~\ref{lem:r-error}). Compared to the corresponding result for Markov games (see~\cite{zhang2023markov}, Lemma~4), we highlight the additional bias terms $(1-\rho)^{m-1}$ and $N^2 H(1-\rho)^{m-1}$, that stem from the superstate approximation error.

For the full proof of Lemma~\ref{lem:est-prob} and its supporting results, we refer to Appendix~\ref{app:proof-results}.

It is standard for sample-based single- and multi-agent reinforcement learning methods to require certain assumptions on the state visitation distribution~\cite{agarwal2021theory,leonardos2022global,ding2022independent}. In the fully observable Markov game case, \cite{zhang2023markov} impose a uniform lower bound on the visitation of the latent states (Assumption~1 in~\cite{zhang2023markov}). The analogous uniform lower bound on the superstate occupancy would be exponentially small in $m$, since the occupancy is spread over $|\mathcal{H}_i^m|$ windows, and is thus incompatible with the choice of $m$ in Theorem~\ref{thm:main}. We therefore only require that the occupancy of the \emph{iterates} covers that of the best response, a weaker condition that is standard for policy optimization methods with local, occupancy-weighted updates such as policy iteration and policy gradient~\cite{kakade2002approximately,agarwal2021theory}, including in finite-memory POMDPs~\cite{cayci2024finite}. In the following, since transitions are decoupled, $d^{m,\pi}_{i,h}$ depends on $\pi$ only through $\pi_i$, and we write $d^{m,\pi_i}_{i,h}$ accordingly.
\begin{assumption}[Superstate concentrability]
\label{ass:concentr}
For iteration $k \in [K]$ of Algorithm~\ref{alg:main}, let
\begin{align*}
\kappa^{(k)} \coloneqq \max_{i \in \mathcal{N},\, h \in [H],\, w_i \in \mathcal{H}_i^{\leq m}} \frac{d^{m,\pi_i^{\star,k}}_{i,h}(w_i)}{d^{m,\pi_i^{(k+1)}}_{i,h}(w_i)},
\qquad \text{where }\;
\pi_i^{\star,k} \in \arg\max_{\pi_i^{\prime} \in \Pi^m_i} V_i(\pi_i^{\prime},\pi^{(k)}_{-i}).
\end{align*}
We assume that there exists $\kappa \geq 1$ with $\kappa \geq \frac{\sum_{k \in [K]} \eta^{(k)}}{\sum_{k \in [K]} \eta^{(k)} / \kappa^{(k)}}$.
\end{assumption}

Each $\kappa^{(k)}$ compares, for every player, the best response occupancy $d^{m,\pi_i^{\star,k}}_{i,h}(w_i)$ against the occupancy of policy $d^{\smash{m,\pi_i^{(k+1)}}}_{i,h}(w_i)$. Rather than bounding each $\kappa^{(k)}$ individually, we only require $\kappa$ to upper bound their $\eta^{(k)}$-weighted mean. This suffices since iteration $k$ contributes an improvement proportional to $\eta^{(k)}/\kappa^{(k)}$, so that only $\sum_{k \in [K]} \eta^{(k)}/\kappa^{(k)} \geq \kappa^{-1}\sum_{k \in [K]} \eta^{(k)}$ enters the telescoping argument in the proof of Lemma~\ref{lem:soft-policy-iter}. 

Next, we state our guarantee for Algorithm~\ref{alg:main} in learning an approximate Nash equilibrium.
\begin{theorem}
\label{thm:main}
Let Assumptions~\ref{ass:potential},~\ref{ass:filter-stab},~\ref{ass:uni-obs}, and~\ref{ass:concentr} hold. For $\epsilon > 0$, $\delta > 0$, suppose each player $i \in \mathcal{N}$ runs Algorithm~\ref{alg:main} simultaneously with stepsizes $\eta^{(k)} = \Theta(1 / \sqrt{N^2 H^3 k})$, exploration $\zeta = \epsilon$, and
\begin{align*}
m &\geq \min \Big\{ H, \; c_1 \cdot \rho^{-1} \log \left( \frac{\kappa NHAO}{\beta \epsilon} \right) \Big\},\\
T &\geq c_2 \cdot \frac{\kappa^2 A^{2m+2} H^6 O^2 N^2 }{\beta^{2m}\epsilon^{2m+2}} \log(1 / \delta),\\
K &\geq c_3 \cdot \frac{\kappa^2 H^5 N^4 A^2}{\epsilon^2},
\end{align*}
where $c_1,c_2,c_3 \in \mathbb{N}$ are constants independent of the problem parameters, and $A \coloneqq \max_{i \in \mathcal{N}} |\mathcal{A}_i|$, $O \coloneqq \max_{i \in \mathcal{N}} |\mathcal{O}_i|$. Then, with probability $1-\delta$, there exists $k \in [K]$ such that the joint policy $\pi^{(k)}$ is an $\epsilon$-approximate Nash equilibrium of the POMG $\mathcal{P}$. The sample and computational complexity of Algorithm~\ref{alg:main} are at most
\begin{align}
\label{eqn:sample-com}
\mathcal{O} \left( \left(\frac{\kappa N H A O}{\beta\epsilon}\right)^{\mathcal{O}(\rho^{-1}\log(A/(\beta \epsilon)))} \log(1/\delta) \right).
\end{align}
\end{theorem}

\paragraph{Proof overview.} Although the superstate game need not admit an exact potential, it satisfies a near-potential property (Proposition~\ref{prop:alpha-potential}). This ensures that each iteration of soft policy iteration yields an approximate improvement, up to an error on the order of $H^2(1-\rho)^{m-1}$ (Lemma~\ref{lem:soft-policy-iter}). Choosing $m$ appropriately controls the accumulated error and yields convergence. For the full proof, we refer to Appendix~\ref{app:proof-results}.

\paragraph{(Quasi-)efficiency.} We highlight that the sample and computational complexity~(\ref{eqn:sample-com}) is quasi-polynomial whenever the concentrability constant $\kappa$ is, and in particular does \emph{not} scale exponentially in the number of players $N$. By decoupledness, $d^{m,\pi}_{i,h}$ depends on $\pi$ only through $\pi_i$, so that $\kappa$ is a single-player quantity carrying no dependence on $N$. Thus, in our decoupled potential game setting, Algorithm~\ref{alg:main} breaks the curse of multi-agency.

\paragraph{Window size tradeoff.} The choice of the window size $m$ induces a tradeoff between approximation accuracy and statistical efficiency. Larger values of $m$ reduce the approximation error of the value in the superstate Markov game relative to the respective value in the original POMG (see Propositions~\ref{prop:approx-ne} and~\ref{prop:alpha-potential}), but increase the sampling complexity required to accurately estimate the superstate model (see Lemma~\ref{lem:est-prob}). Choosing $m$ on a logarithmic scale as in Theorem~\ref{thm:main} balances these effects.

\paragraph{Comparison with prior work.} Compared to~\cite{liu2022sample}, our framework relies on a different structural assumption, filter stability rather than observability, and focuses on decoupled dynamics with potential structure. This restriction enables an \emph{independent learning} approach that does not rely on computationally intractable oracles. In contrast to~\cite{liu2023partially}, our setting does not require information sharing. Moreover, unlike~\cite{liu2022sample,liu2023partially} our method does not suffer from the curse of multi-agency.

\section{Conclusion}
\label{sec:conclusion}

We studied Nash equilibrium learning in partially observable Markov games under decoupled dynamics and potential structure. We proposed an independent learning algorithm that converges to an approximate Nash equilibrium with quasi-polynomial sample and computational complexity, under a boundedness assumption on a superstate concentrability coefficient. Our analysis is based on a superstate Markov game representation, which we show approximately preserves equilibria and admits a near-potential structure.

\paragraph{Limitations and future work.} The setting of POMGs with decoupled dynamics and potential structure forms a rich but still restricted subclass of POMGs. Understanding the boundary between tractable and intractable regimes remains an important direction for future work. Another limitation is Assumption~\ref{ass:concentr}, which bounds the superstate concentrability coefficients along the trajectory of Algorithm~\ref{alg:main}; identifying concrete classes of POMGs in which $\kappa$ is small is an interesting direction. To the best of our knowledge, this is the first work on provable independent learning in POMGs. Our results open up new possibilities for independent learning in partially observable multi-agent systems and provide a starting point for further extensions to other game classes, such as two-player zero-sum games, and to broader equilibrium notions such as correlated equilibria in general-sum games.

\bibliographystyle{unsrtnat}
\bibliography{refs}

\newpage
\appendix

\section*{\LARGE{Supplementary Material}}
\addtocontents{toc}{\protect\setcounter{tocdepth}{2}}
\tableofcontents
\bigskip
\noindent\rule{\textwidth}{0.4pt}
\newpage

\section{Overview of Notation}
\label{app:notation}
Table~\ref{tbl:notation} gives an overview of the notation used throughout this paper. All notations are introduced before their first use as well.
\begin{table}[H]
\caption{Overview of notation}
\label{tbl:notation}
\begin{center}
\small
\begin{tabular}{ll}
\toprule
\textbf{Symbol} & \textbf{Description} \\
\midrule
\multicolumn{2}{c}{\textit{POMG}} \\
\midrule
$\mathcal{P} = (\mathcal{N},H,\mathcal{S},\mathcal{A},\mathcal{O},r,\mathbb{P},\mathbb{O},\mu)$ & the partially observable Markov game \\
$\mathcal{N} = [N]$ & set of $N$ players \\
$H$ & planning horizon \\
$\mathcal{S} = \prod_i \mathcal{S}_i$ & joint state space \\
$\mathcal{A} = \prod_i \mathcal{A}_i$ & joint action space; $A \coloneqq \max_i |\mathcal{A}_i|$ \\
$\mathcal{O} = \prod_i \mathcal{O}_i$ & joint observation space; $O \coloneqq \max_i |\mathcal{O}_i|$ \\
$r_{i,h}: \mathcal{S} \times \mathcal{A} \to [0,1]$ & reward function for player $i$ at step $h$ \\
$\mathbb{P}_{i,h}(\cdot \mid s_i, a_i)$ & transition kernel for player $i$ at step $h$ \\
$\mathbb{O}_{i,h}(\cdot \mid s_i)$ & observation kernel for player $i$ at step $h$ \\
$\mu = \prod_i \mu_i$ & initial state distribution over $\mathcal{S}$ \\
\midrule
\multicolumn{2}{c}{\textit{Histories and Policies}} \\
\midrule
$\mathcal{H}_i^h = (\mathcal{O}_i \times \mathcal{A}_i)^{h-1} \times \mathcal{O}_i$ & set of $h$-step local histories for player $i$ \\
$\mathcal{H}_i^{\le h}$ & $\bigcup_{h'=0}^{h} \mathcal{H}_i^{h'}$ \\
$\tau_i = (o_{i,1},a_{i,1},\dots,a_{i,h-1},o_{i,h})$ & local history of length $|\tau_i| = h$; $\circ$ denotes concatenation \\
$\tau_{i,k:l}$ & sub-history from step $k$ to step $l$ (inclusive) \\
$\Pi_i^m$ & $m$-step finite-window policies for player $i$ \\
$V_i(\pi)$ & value of player $i$ under joint policy $\pi$ in $\mathcal{P}$ \\
\midrule
\multicolumn{2}{c}{\textit{Superstate Markov Game $\mathcal{G}^m$}} \\
\midrule
$w_i \in \mathcal{H}_i^{\le m}$ & local superstate (truncated history window) of player $i$ \\
$b_{i,h}^m(s_i \mid w_i)$ & player $i$'s belief over $\mathcal{S}_i$ induced by window $w_i$ at step $h$ \\
$b_{i,h}(s_i \mid \tau_i)$ & player $i$'s belief over $\mathcal{S}_i$ given full history $\tau_i$ \\
$\mathbb{P}_h^m(w' \mid w, a)$ & superstate transition kernel of $\mathcal{G}^m$ \\
$r_{i,h}^m(w, a)$ & reward of player $i$ in $\mathcal{G}^m$ \\
$V_i^m(\pi)$ & value of player $i$ under $\pi$ in $\mathcal{G}^m$ \\
$\overline{Q}_{i,h}^{m,\pi}(w_i, a_i)$ & marginalized $Q$-function in $\mathcal{G}^m$ under policy $\pi$ \\
$d_{i,h}^{m,\pi}(w_i)$ & superstate visitation distribution at step $h$ under $\pi$ \\
\midrule
\multicolumn{2}{c}{\textit{Potential Structure}} \\
\midrule
$\lambda_{i,h}^{\pi_i}(s_i,a_i)$ & state-action occupancy measure of player $i$ \\
$\Phi: \Pi \to \mathbb{R}$ & potential function of the underlying Markov game \\
$\phi_h: \mathcal{S} \times \mathcal{A} \to \mathbb{R}$ & stage-wise potential function at step $h$ \\
$\Psi: \Pi^m \to \mathbb{R}$ & near-potential function for $\mathcal{G}^m$ (see Proposition~\ref{prop:alpha-potential}) \\
\midrule
\multicolumn{2}{c}{\textit{Assumptions and Key Parameters}} \\
\midrule
$\rho > 0$ & filter stability contraction factor (see Assumption~\ref{ass:filter-stab}) \\
$\beta > 0$ & minimum observation probability (see Assumption~\ref{ass:uni-obs}) \\
$\kappa \geq 1$ & superstate concentrability constant (see Assumption~\ref{ass:concentr}) \\
$\zeta \in (0,1]$ & exploration probability of Algorithm~\ref{alg:main} \\
$\epsilon_\rho^m \coloneqq 4H^2(1-\rho)^{m-1}$ & finite-window approximation error (see Lemma~\ref{lem:value-approx}) \\
\midrule
\multicolumn{2}{c}{\textit{Algorithm~\ref{alg:main}}} \\
\midrule
$T$ & number of episodes per iteration \\
$K$ & number of iterations \\
$\eta^{(k)}$ & stepsize at iteration $k$ \\
$\hat{\mathbb{P}}_{i,h}^{m,\pi_i},\ \hat{r}_{i,h}^{m,\pi_{-i}}$ & empirical estimates of superstate transitions and rewards \\
$\hat{Q}_{i,h}^{(k)}$ & estimated $Q$-function at iteration $k$ \\
\bottomrule
\end{tabular}
\end{center}
\vskip -0.1in
\end{table}
\newpage

\section{Discussion of decoupled dynamics assumption}
\label{app:decoupled}

\subsection{Illustrative examples}

We provide several representative application domains in which the assumption of decoupled transition dynamics arises naturally, while interactions are captured through the reward structure.

\paragraph{Wireless networks.}
Consider a wireless communication network in which $N$ users aim to transmit packets over a shared channel. Each user maintains a local queue of packets and, at each time step, selects a transmission power level. The local state corresponds to the queue length, which evolves based on the transmission success and exogenous packet arrivals. Crucially, under standard modeling assumptions, the queue dynamics depend only on the user's own state and action, and are thus conditionally independent across users.

In contrast, rewards depend on the joint action profile: simultaneous high-power transmissions lead to increased interference and congestion, reducing transmission success probabilities. This induces a coupling in the reward function, often leading to a (statewise) congestion (and thus potential) structure. Such models have been studied in the networking literature (see, e.g.,~\cite{altman2007constrained,altman2009dynamic}). In practice, the queue length may only be partially observed due to delays, noise, or aggregation effects, leading naturally to a partially observable setting.

\paragraph{Distributed energy markets.}
In distributed energy systems~\cite{etesami2018stochastic}, multiple prosumers (e.g., households equipped with photovoltaic panels and storage) interact through a shared energy market. Each prosumer's local state is given by its battery level, which evolves according to locally generated and consumed energy, as well as charging and discharging decisions. Assuming participants are geographically separated with independent local conditions, these state dynamics may be modelled as decoupled across agents.

However, the reward, typically representing negative energy cost, depends on the aggregate supply and demand in the market, and hence on the joint actions of all agents. This induces strategic coupling through prices or clearing mechanisms, while maintaining independent state evolution.

\paragraph{Multi-robot systems.}
Decoupled transition dynamics also naturally arise in multi-robot systems and autonomous driving. Each agent (robot or vehicle) evolves according to its own local dynamics, which depend primarily on its state and control inputs.
Coupling between agents typically enters through the reward function, which encodes coordination objectives such as collision avoidance, or task allocation. In many practical settings, agents have only partial observations of the local state due to limited sensing.
For instance, \cite{xiao2025asynchronous} implements reinforcement learning methods under partial observability in real-world multi-robot tasks. Here, agents learn localized policies based on their individual state observations.

\section{Proofs for Section~\ref{sec:mg-repr}}
\label{app:proof-mg-repr}

In this section, we provide the proofs for our structural results on POMGs presented in Section~\ref{sec:mg-repr}.

\subsection{Finite window approximation results}

In this subsection, we establish several approximation results under filter stability that lead us towards proving Proposition~\ref{prop:approx-ne}.

First, we show that beliefs formed from joint histories corresponding to the same superstate must be close to each other in terms of total variation distance.
\begin{lemma}
\label{lem:beliefs-tv}
For $m,h,h^{\prime} \in \mathbb{N}$ with $h,h^{\prime} \geq m$, let $\tau \in \mathcal{H}^h, \tau^{\prime} \in \mathcal{H}^{h^{\prime}}$ such that $\tau_{h-m+1:h}=\tau^{\prime}_{h^{\prime}-m+1:h^{\prime}}$. Then, under Assumption~\ref{ass:filter-stab}, we have
\begin{align*}
\left\lVert b_h(\cdot \mid \tau)-b_{h^{\prime}}(\cdot \mid \tau^{\prime}) \right\rVert_{TV} \leq (1-\rho)^{m-1}.
\end{align*}
\end{lemma}
\begin{proof}
Since the POMG can be viewed as an POMDP with action space $\mathcal{A}$ and observation space $\mathcal{O}$, the result follows equivalently to the respective single-agent result, see Lemma~1 in \cite{anjarlekarscalable}. An $m$-step window consists of one observation followed by $m-1$ appended action-observation pairs, so Assumption~\ref{ass:filter-stab} is applied $m-1$ times, and hence the desired bound follows.
\end{proof}

\paragraph{Remark.} Although Assumption~\ref{ass:filter-stab} is stated for the joint belief, it hides no exponential dependence on the number of players. It enters our analysis only through the belief bound of Lemma~\ref{lem:beliefs-tv}. If filter stability is instead assumed playerwise with parameter $\rho_i$ for each $i \in \mathcal{N}$, setting $\rho \coloneqq \min_{i \in \mathcal{N}} \rho_i$, using decoupledness and the fact that $\lVert \prod_{i} P_i - \prod_{i} Q_i \rVert_{TV} \leq \sum_{i} \lVert P_i-Q_i \rVert_{TV}$ for product measures, a linear $N$-dependence is introduced which gets absorbed in the big-$\mathcal{O}$ notation of Theorem~\ref{thm:main}.

Next, we bound the error in terms of superstate rewards and transition kernels resulting from finite-window approximation. Recall that $\circ$ denotes concatenation of histories.
\begin{lemma}
\label{lem:aux-p-r-err}
Under Assumption~\ref{ass:filter-stab}, let $h \geq m$ and $w \in \mathcal{H}^m$, $\tau \in \mathcal{H}^h$ with $w=\tau_{h-m+1:h}$. Then, for any $i \in \mathcal{N}$, $h \in [H]$, and $a \in \mathcal{A}$,
\begin{align*}
\left\vert r^m_{i,h}(w,a)-r^H_{i,h}(\tau,a) \right\vert &\leq 2(1-\rho)^{m-1},\\
\sum_{o \in \mathcal{O}} \left\vert \mathbb{P}^m_h(w_{2:m} \circ (a,o) \mid w,a) - \mathbb{P}^H_h(\tau_{h-m+1:h} \circ (a,o) \mid \tau,a) \right\vert &\leq 2(1-\rho)^{m-1}.
\end{align*}
\end{lemma}
\begin{proof}
For the first bound, note that
\begin{align*}
\left\vert r^m_{i,h}(w,a)-r^H_{i,h}(\tau,a) \right\vert
\leq \sum_{s \in \mathcal{S}} r_{i,h}(s,a) \left\vert b^m_h(s \mid w)-b_h(s \mid \tau) \right\vert
\leq 2(1-\rho)^{m-1}
\end{align*}
where in the last step we have used Lemma~\ref{lem:beliefs-tv} and the fact that rewards are bounded by $1$.

For the second bound, denoting $w^{\prime}_o \coloneqq w_{2:m} \circ (a,o)$ and $\tau^{\prime}_o \coloneqq \tau \circ (a,o)$, we have
\begin{align*}
&\sum_{o \in \mathcal{O}} \left\vert \mathbb{P}^m_h( w^{\prime}_o \mid w,a) - \mathbb{P}^H_h( \tau^{\prime}_o \mid \tau,a) \right\vert \\
&\qquad= \sum_{o \in \mathcal{O}} \sum_{s,s^{\prime} \in \mathcal{S}} \mathbb{O}_{h+1}(o \mid s^{\prime}) \mathbb{P}_h(s^{\prime} \mid s,a) \left\vert b^m_h(s \mid w)-b_h(s \mid \tau) \right\vert \\
&\qquad\leq \sum_{s \in \mathcal{S}} \left\vert b^m_h(s \mid w)-b_h(s \mid \tau) \right\vert \sum_{o \in \mathcal{O}} \sum_{s^{\prime} \in \mathcal{S}} \mathbb{O}_{h+1}(o \mid s^{\prime}) \mathbb{P}_h(s^{\prime} \mid s,a) \\
&\qquad\leq 2(1-\rho)^{m-1}.
\end{align*}
\end{proof}

Using the above bound, we obtain the following approximation result for value functions and best responses, restated here for convenience.
\valueapproxlemma*
\begin{proof}
The proof of the first bound proceeds by showing a stronger inequality by backward induction on $h$. For this, we define for any $w \in \mathcal{H}^{\leq m}$ and $\tau \in \mathcal{H}^{h}$ the step $h$ value functions
\begin{align*}
V_{i,h}(\pi; \tau) &\coloneqq \mathbb{E}_{\pi,\, s_1 \sim \mu}\left[ \sum_{h^{\prime}=h}^H r^m_{i,h^{\prime}}(s_{h^{\prime}},a_{h^{\prime}}) \mid (o_1,a_1,\dots,a_{h-1},o_h)=\tau \right] \\
V^m_{i,h}(\pi;w) &\coloneqq \mathbb{E}_{\pi,\, s_1 \sim \mu}\left[ \sum_{h^{\prime}=h}^H r^m_{i,h^{\prime}}(w_{h^{\prime}},a_{h^{\prime}}) \mid w_{h}=w \right].
\end{align*}
The inequality we aim to show for all $h \in [H]$ is
\begin{align}
\label{eqn:ind-ineq}
\left\vert V^m_{i,h}(\pi;w) - V_{i,h}(\pi;\tau) \right\vert
\leq 4(1-\rho)^{m-1} H(H+1-h).
\end{align}
For $H+1$, (\ref{eqn:ind-ineq}) trivially holds. Suppose (\ref{eqn:ind-ineq}) holds for some $h \in [H]$. The Bellman expectation equations are given by
\begin{align*}
V_{i,h}(\pi^{\prime};\tau) &= \sum_{a \in \mathcal{A}} \pi^{\prime}_h(a \mid \tau) \left[ r^H_{i,h}(\tau,a) + \sum_{o \in \mathcal{O}} \mathbb{P}^H_h(\tau \circ (a,o) \mid \tau,a) V_{i,h+1}(\pi^{\prime}; \tau \circ (a,o)) \right], \\
V^m_{i,h}(\pi;w) &= \sum_{a \in \mathcal{A}} \pi_h(a \mid w) \left[ r^m_{i,h}(w,a) + \sum_{o \in \mathcal{O}} \mathbb{P}^m_h(w_{2:m} \circ (a,o) \mid w,a) V^m_{i,h+1}(\pi; w_{2:m} \circ (a,o)) \right].
\end{align*}
Hence, we can write the difference as
\begin{align*}
&\left\vert V^m_{i,h}(\pi;w) - V_{i,h}(\pi^{\prime};\tau) \right\vert \\
&\qquad \leq \sum_{a \in \mathcal{A}} \pi_h(a \mid w) \Bigg[ \left\vert r^m_{i,h}(w,a)-r^H_{i,h}(\tau,a) \right\vert \\
&\qquad\qquad+ \sum_{o \in \mathcal{O}} \mathbb{P}^m_h(w_{2:m} \circ (a,o) \mid w,a) \left[ V^m_{i,h+1}(\pi; w_{2:m} \circ (a,o)) - V_{i,h+1}(\pi^{\prime}; \tau \circ (a,o)) \right] \\
&\qquad\qquad+ \sum_{o \in \mathcal{O}} \left[ \mathbb{P}^m_h(w_{2:m} \circ (a,o) \mid w,a) - \mathbb{P}^H_h(\tau \circ (a,o) \mid \tau,a) \right] V_{i,h+1}(\pi^{\prime}; \tau \circ (a,o)) \Bigg] \\
&\qquad \overset{(a)}\leq \sum_{a \in \mathcal{A}} \pi_h(a \mid w) \Bigg[ 2(1-\rho)^{m-1} + \max_{o \in \mathcal{O}} \left\vert V^m_{i,h+1}(\pi; w_{2:m} \circ (a,o)) - V_{i,h+1}(\pi^{\prime}; \tau \circ (a,o)) \right\vert + 2(1-\rho)^{m-1} H \Bigg] \\
&\qquad \leq 4H(1-\rho)^{m-1} + \max_{a \in \mathcal{A}, o \in \mathcal{O}} \left\vert V^m_{i,h+1}(\pi; w_{2:m} \circ (a,o)) - V_{i,h+1}(\pi^{\prime}; \tau \circ (a,o)) \right\vert \\
&\qquad \overset{(b)}\leq 4H(1-\rho)^{m-1} + 4(1-\rho)^{m-1} H(H-h) \\
&\qquad \leq 4(1-\rho)^{m-1} H(H+1-h)
\end{align*}
where in (a) we use Lemma~\ref{lem:aux-p-r-err}, and in (b) we apply the induction hypothesis.

To show the second bound, define the best response functions for any $i \in \mathcal{N}$, $h \in [H]$, $w \in \mathcal{H}^{\leq m}$, and $\tau \in \mathcal{H}^{h}$,
\begin{align*}
V^{\dagger}_{i,h}(\pi^{\prime}_{-i};\tau) &\coloneqq \max_{\hat{\pi}_i^{\prime} \in \Pi^H_i}V_{i,h}(\hat{\pi}_i^{\prime},\pi_{-i};\tau), \\
V_{i,h}^{m,\dagger}(\pi_{-i};w) &\coloneqq \max_{\hat{\pi}_i \in \Pi^m_i}V^m_{i,h}(\hat{\pi}_i,\pi_{-i};w).
\end{align*}
By Bellman optimality, we have
\begin{align*}
&V_{i,h}^{\dagger}(\pi^{\prime}_{-i};\tau) = \max_{a_i \in \mathcal{A}_i} \sum_{a_{-i} \in \mathcal{A}_{-i}} \pi^{\prime}_{-i}(a_{-i} \mid \tau_{-i}) \Bigg[ r^H_{i,h}(\tau,a_i,a_{-i}) \\
&\qquad\qquad\qquad\qquad+ \sum_{o \in \mathcal{O}} \mathbb{P}^H_h(\tau \circ (a_i,a_{-i},o) \mid \tau,a_i,a_{-i}) V_{i,h+1}^{\dagger}(\pi^{\prime}_{-i}; \tau \circ (a_i,a_{-i},o)) \Bigg], \\
&V_{i,h}^{m,\dagger}(\pi_{-i};w) = \max_{a_i \in \mathcal{A}_i} \sum_{a_{-i} \in \mathcal{A}_{-i}} \pi_{-i}(a_{-i} \mid w_{-i}) \Bigg[ r^m_{i,h}(w,a_i,a_{-i}) \\
&\qquad\qquad\qquad\qquad+ \sum_{o \in \mathcal{O}} \mathbb{P}^m_h(w_{2:m} \circ (a_i,a_{-i},o) \mid w,a_i,a_{-i}) V_{i,h+1}^{m,\dagger}(\pi_{-i};w_{2:m} \circ (a_i,a_{-i},o)) \Bigg].
\end{align*}
Therefore, using the same steps as above, we can decompose the difference as above and apply Lemma~\ref{lem:aux-p-r-err} to obtain
\begin{align*}
&\left\vert V_{i,h}^{m,\dagger}(\pi_{-i};w) - V_{i,h}^{\dagger}(\pi^{\prime}_{-i};\tau) \right\vert \\
&\qquad \leq \max_{a_i \in \mathcal{A}_i} \sum_{a_{-i} \in \mathcal{A}_{-i}} \pi_{-i}(a_{-i} \mid w_{-i}) \Bigg[ 4H(1-\rho)^{m-1} \\
&\qquad\qquad+ \max_{o \in \mathcal{O}} \left\vert V_{i,h+1}^{m,\dagger}(\pi_{-i};w_{2:m} \circ (a_i,a_{-i},o)) - V_{i,h+1}^{\dagger}(\pi^{\prime}_{-i}; \tau \circ (a_i,a_{-i},o)) \right\vert \Bigg] \\
&\qquad \leq 4H(1-\rho)^{m-1} + \max_{a \in \mathcal{A},o \in \mathcal{O}} \left\vert V_{i,h+1}^{m,\dagger}(\pi_{-i};w_{2:m} \circ (a_i,a_{-i},o)) - V_{i,h+1}^{\dagger}(\pi^{\prime}_{-i}; \tau \circ (a_i,a_{-i},o)) \right\vert.
\end{align*}
Then, by an induction argument over $h$ similar to above, we can show that for any $h \in [H]$,
\begin{align*}
\left\vert V_{i,h}^{m,\dagger}(\pi_{-i};w) - V_{i,h}^{\dagger}(\pi^{\prime}_{-i};\tau) \right\vert
\leq 4(1-\rho)^{m-1} H(H+1-h),
\end{align*}
which implies the desired bound.
\end{proof}

\subsection{Nash equilibrium approximation}

As a direct application of Lemma~\ref{lem:value-approx}, we obtain the following approximation guarantee for Nash equilibria for $m$-step history-dependent policies.
\approxneprop*
\begin{proof}
With best response function defined as in the proof of Lemma~\ref{lem:value-approx}, we have
\begin{align*}
V^{\dagger}_i(\pi_{-i}) - V_i(\pi)
&\leq \left[ V^{\dagger}_i(\pi_{-i}) - V_i^{m,\dagger}(\pi_{-i}) \right] + \left[ V_i^{m,\dagger}(\pi_{-i}) - V_i^m(\pi) \right] + \Big[ V_i^m(\pi) - V_i(\pi) \Big] \\
&\overset{(a)}{\leq} \left\vert V^{\dagger}_i(\pi_{-i}) - V_i^{m,\dagger}(\pi_{-i}) \right\vert + \epsilon + \Big\vert V_i^m(\pi) - V_i(\pi) \Big\vert \\
&\overset{(b)}{\leq} \epsilon + 2\epsilon^{m}_\rho
\end{align*}
where (a) is due to $\pi$ being an $\epsilon$-approximate Nash equilibrium, and (b) follows from the two respective bounds in Lemma~\ref{lem:value-approx}.
\end{proof}

\subsection{Near-potential structure}

Having shown our Nash equilibrium approximation guarantee, we next establish near-potential structure of the superstate Markov game $\mathcal{G}^m$ in terms of the value functions $V_i^m$, for $i \in \mathcal{N}$.

\alphapotentialprop*
\begin{proof}
The proof is divided into two parts. First, using decoupledness of transitions, we show that under Assumption~\ref{ass:potential} on the underlying latent state Markov game, there exists an exact potential function for the POMG. As a second step, based on this, we then prove the $(2\epsilon_\rho^m)$-approximate potential structure for the superstate Markov game.
\begin{itemize}
\item \textbf{POMG is potential:} For any $m$-step policy $\pi_i \in \Pi^m_i$ and $h \in [H]$, let $\lambda^{\pi_i}_{i,h} \in \Delta(\mathcal{S}_i \times \mathcal{A}_i)$ denote the induced playerwise state-action visitation probability at step $h$, defined as
\begin{align*}
\lambda_{i,h}^{\pi_i}(s_i,a_i) \coloneqq P_{\pi_i} \left( s_{i,h}=s_i, a_{i,h}=a_i \right)
\end{align*}
where $P_{\pi_i}$ refers to the distribution over player $i$'s (unobserved) state trajectories when following $\pi_i$ in the underlying POMG. Analogously, for a joint policy $\pi \in \Pi^m$, let $\lambda_h^{\pi}(s,a) \coloneqq P_{\pi}(s_h=s, a_h=a) \in \Delta(\mathcal{S} \times \mathcal{A})$ denote the induced joint state-action occupancy at step $h$. Since the initial distribution $\mu$, the transitions $\mathbb{P}$, the observations $\mathbb{O}$, and the product policy $\pi$ all factor across players, each player's trajectory evolves independently, so the joint occupancy factorizes into the playerwise marginals, i.e.,
\begin{align*}
\lambda_h^{\pi}(s,a) = \prod_{j \in \mathcal{N}} \lambda_{j,h}^{\pi_j}(a_i,s_i)
\end{align*}
for any $a \in \mathcal{A}$, $s \in \mathcal{S}$, and $h \in [H]$.

By Assumption~\ref{ass:potential}, there exist $\phi_h:\mathcal{S} \times \mathcal{A} \to \mathbb{R}$ for each $h \in [H]$ such that for all $i \in \mathcal{N}$, $h \in [H]$, $s \in \mathcal{S}$, $s_i^{\prime} \in \mathcal{S}_i$, $a \in \mathcal{A}$, and $a_i^{\prime} \in \mathcal{A}_i$, we have
\begin{align*}
&r_{i,h}(s_{i,h}', s_{-i,h}, a_{i,h}', a_{-i,h}) - r_{i,h}(s_{i,h}, s_{-i,h}, a_{i,h}, a_{-i,h})\\
&\qquad\qquad = \phi_{h}(s_{i,h}', s_{-i,h}, a_{i,h}', a_{-i,h}) - \phi_{h}(s_{i,h}, s_{-i,h}, a_{i,h}, a_{-i,h}).
\end{align*}
Define the function $\Phi: \Pi^m \to \mathbb{R}$ as
\begin{align*}
\Phi(\pi) \coloneqq \mathbb{E}_{\pi,s_1 \sim \mu} \left[ \sum_{h=1}^H \phi_h(s_h,a_h) \right],
\end{align*}
where the expectation is taken over the distribution of trajectories induced by the POMG $\mathcal{P}$ under joint policy $\pi \in \Pi^m$, including randomness from state transitions according to $\mathbb{P}$ and observations emissions according to $\mathbb{O}$.

Further define for any $i \in \mathcal{N}$, $\pi_{-i} \in \Pi_{-i}^m$, $h \in [H]$, $s_i \in \mathcal{S}_i$, and $a_i \in \mathcal{A}_i$,
\begin{align*}
r_{i,h}^{\pi_{-i}}(s_i,a_i) &\coloneqq \sum_{s_{-i},a_{-i}} r_{i,h}(s_i,s_{-i},a_i,a_{-i}) \prod_{j \in \mathcal{N} \setminus \left\{ i \right\}} \lambda_{j,h}^{\pi_j}(s_j,a_j),\\
\phi_h^{\pi_{-i}}(s_i,a_i) &\coloneqq \sum_{s_{-i},a_{-i}} \phi_h(s_i,s_{-i},a_i,a_{-i}) \prod_{j \in \mathcal{N} \setminus \left\{ i \right\}} \lambda_{j,h}^{\pi_j}(s_j,a_j).
\end{align*}
Then, for any $m$-step policy $\pi \in \Pi^m$ and deviation $\pi_i^{\prime} \in \Pi_i^m$, with the POMG value function defined as in (\ref{eqn:value-def}), we can write
\begin{align*}
&V_i({\pi}) - V_i(\pi_i^{\prime},\pi_{-i}) \\[5pt]
&\overset{(a)}{=} \sum_{h \in [H]} r^{\top}_{i,h} \lambda_h^{\pi} - \sum_{h \in [H]} r^{\top}_{i,h} \lambda_h^{(\pi_i^{\prime},\pi_{-i})} \\
&\overset{(b)}{=} \sum_{h \in [H]} r^{\top}_{i,h} \left( \bigtimes_{j \in \mathcal{N}} \lambda_{j,h}^{\pi_j} \right) - \sum_{h \in [H]} r^{\top}_{i,h} \left( \lambda_{i,h}^{\pi_i^{\prime}} \times \left( \bigtimes_{j \in \mathcal{N} \setminus \left\{ i \right\}} \lambda_{j,h}^{\pi_j} \right) \right) \\
&= \sum_{h \in [H]} \left( r_{i,h}^{\pi_{-i}} \right)^{\top} \lambda_{i,h}^{\pi_i} - \left( r_{i,h}^{\pi_{-i}} \right)^{\top} \lambda_{i,h}^{\pi_i^{\prime}} \\
&\overset{(c)}{=} \frac{1}{2} \sum_{h \in [H]} \sum_{s_i,s_i^{\prime},a_i,a_i^{\prime}} \left( \lambda_{i,h}^{\pi_i}(s_i,a_i) \lambda_{i,h}^{\pi_i^{\prime}}(s_i^{\prime},a_i^{\prime}) - \lambda_{i,h}^{\pi_i}(s^{\prime}_i,a^{\prime}_i) \lambda_{i,h}^{\pi_i^{\prime}}(s_i,a_i) \right) \left( r^{\pi_{-i}}_{i,h}(s_i,a_i) - r^{\pi_{-i}}_{i,h}(s^{\prime}_i,a^{\prime}_i) \right) \\
&= \frac{1}{2} \sum_{h \in [H]} \sum_{s_i,s_i^{\prime},a_i,a_i^{\prime}} \left( \lambda_{i,h}^{\pi_i}(s_i,a_i) \lambda_{i,h}^{\pi_i^{\prime}}(s_i^{\prime},a_i^{\prime}) - \lambda_{i,h}^{\pi_i}(s^{\prime}_i,a^{\prime}_i) \lambda_{i,h}^{\pi_i^{\prime}}(s_i,a_i) \right) \left( \phi_h^{\pi_{-i}}(s_i,a_i) - \phi_h^{\pi_{-i}}(s^{\prime}_i,a^{\prime}_i) \right) \\
&\,= \sum_{h \in [H]} \left( \phi_h^{\pi_{-i}} \right)^{\top} \lambda_{i,h}^{\pi_i} - \left( \phi_h^{\pi_{-i}} \right)^{\top} \lambda_{i,h}^{\pi_i^{\prime}} \\
&\,= \sum_{h \in [H]} \phi_h^{\top} \left( \bigtimes_{j \in \mathcal{N}} \lambda_{j,h}^{\pi_j} \right) - \phi_h^{\top} \left( \lambda_{i,h}^{\pi_i^{\prime}} \times \left( \bigtimes_{j \in \mathcal{N} \setminus \left\{ i \right\}} \lambda_{j,h}^{\pi_j} \right) \right) \\
&= \Phi({\pi}) - \Phi(\pi_i^{\prime},\pi_{-i})
\end{align*}
where (a) writes each value as a linear form in the joint state-action occupancy, which holds by linearity of expectation since each $r_{i,h}$ depends only on $(s_h,a_h)$; (b) factorizes the joint occupancy into the playerwise marginals, $\lambda_h^{\pi} = \bigtimes_{j \in \mathcal{N}} \lambda_{j,h}^{\pi_j}$, by decoupledness of the initial distribution, transitions, observations, and policies; and (c) uses the identity proven in Lemma~\ref{lem:pairwise-difference}. We have established potential structure with potential function $\Phi$ over the policy class $\Pi^m$ for the POMG $\mathcal{P}$. Next, we use this fact, together with our finite-window approximation guarantees relating the value functions $V_i$ and $V^m_i$, to show near-potential structure for the superstate Markov game.

\item \textbf{Superstate Markov game is near-potential:} For any $\pi \in \Pi^m$, $i \in \mathcal{N}$, and~$\pi_i^{\prime} \in \Pi^m_i$,
\begin{align*}
\left( V_i^m({\pi})-V_i^m(\pi_i^{\prime},\pi_{-i}) \right) &- \left( \Phi(\pi) - \Phi(\pi_i^{\prime},\pi_{-i}) \right) \\[8pt]
&= \underbrace{\left( V_i^m(\pi)-V_i(\pi) \right) + \left( V_i(\pi_i^{\prime},\pi_{-i})-V_i^m(\pi_i^{\prime},\pi_{-i}) \right)}_{(a)} \\
&\quad+ \underbrace{\left( V_i(\pi)-V_i(\pi_i^{\prime},\pi_{-i}) \right) - \left( \Phi(\pi) - \Phi(\pi_i^{\prime},\pi_{-i}) \right)}_{(b)}.
\end{align*}
Note that since $\Phi$ is a potential function of the POMG $\mathcal{P}$, we have $(b)=0$. Moreover, using Lemma~\ref{lem:value-approx}, we bound $\left\vert (a) \right\vert \leq 2\epsilon^{m}_\rho$, which concludes the proof.
\end{itemize}
\end{proof}

\subsection{Auxiliary lemma}

In the above proof of Proposition~\ref{prop:alpha-potential}, we invoked the following algebraic identity, which expresses the difference of two linear forms over distributions as a symmetric pairwise sum. Here we provide its proof.

\begin{lemma}
\label{lem:pairwise-difference}
For $k \in \mathbb{N}$ and $\mathcal{I}=\{1,\dots,k\}$, let $p,q \in \Delta(\mathcal{I})$, and let $x:\mathcal{I}\to\mathbb{R}$ be arbitrary. Then
\begin{align*}
\sum_{i\in \mathcal{I}} p(i)x(i) - \sum_{i\in \mathcal{I}} q(i)x(i)
=
\frac12
\sum_{i,j\in \mathcal{I}}
\bigl(p(i)q(j)-p(j)q(i)\bigr)\,\bigl(x(i)-x(j)\bigr).
\end{align*}
\end{lemma}

\begin{proof}
Expanding the right-hand side yields
\begin{align*}
\frac12\sum_{i,j} p(i)q(j)x(i)
-\frac12\sum_{i,j} p(i)q(j)x(j)
-\frac12\sum_{i,j} p(j)q(i)x(i)
+\frac12\sum_{i,j} p(j)q(i)x(j).
\end{align*}
We first group terms according to whether they multiply $x(i)$ or $x(j)$:
\begin{align*}
\sum_{i,j} \frac12\bigl(p(i)q(j)-p(j)q(i)\bigr)x(i)
+
\sum_{i,j} \frac12\bigl(p(j)q(i)-p(i)q(j)\bigr)x(j).
\end{align*}
For the first sum, fixing $i$ and summing over $j$ yields
\[
\sum_j \frac12\bigl(p(i)q(j)-p(j)q(i)\bigr)
=
\frac12\Bigl(p(i)\sum_j q(j)-q(i)\sum_j p(j)\Bigr)
=
\frac12\bigl(p(i)-q(i)\bigr),
\]
where we used $\sum_j p(j)=\sum_j q(j)=1$. Hence
\[
\sum_{i,j} \frac12\bigl(p(i)q(j)-p(j)q(i)\bigr)x(i)
=
\sum_i \frac12\bigl(p(i)-q(i)\bigr)x(i).
\]
Applying the same argument to the second sum (with indices relabeled) gives
\[
\sum_{i,j} \frac12\bigl(p(j)q(i)-p(i)q(j)\bigr)x(j)
=
\sum_i \frac12\bigl(q(i)-p(i)\bigr)x(i).
\]
Combining the two expressions and canceling terms yields
\[
\sum_i (p(i)-q(i))x(i)
=
\sum_i p(i)x(i)-\sum_i q(i)x(i),
\]
which completes the proof.
\end{proof}

\section{Algorithm details}
\label{app:alg-details}

Below, we provide further details omitted in the main part by formally specifying how to compute empirical frequencies for estimating superstate transition probabilities and rewards based on the sampled action-observation sequences. Denote by $\tau^{(t)}_i=(a^{(t)}_{i,h},o^{(t)}_{i,h},r_{i,h}^{(t)})_{h=1}^{H}$ for $t \in [T]$ the $t$-th sequence sampled at some iteration $k \in [K]$ of Algorithm~\ref{alg:main} when players follow the joint policy~$\pi^{(k)} \in \Pi^m$.

\paragraph{Transition probabilities.} At each player $i \in \mathcal{N}$, at step $h \in [H]$, and for all $w_i,w_i^{\prime} \in \mathcal{H}_i^{\leq m}$ and $a_i \in \mathcal{A}_i$, take the empirical average of transitions from $w_i$ to $w_i^{\prime}$, that is,
\begin{align}
\hat{\mathbb{P}}_{i,h}^{m,\pi_i^{(k)}}(w_i^{\prime} \mid w_i,a_i) \coloneqq \label{eq:sample-p}
\begin{cases}
\frac{\sum_{t=1}^{T} \mathbf{1}\{ \tau^{(t)}_{i, h-|w_i^{\prime}|+2\,:\,h+1}=w_i^{\prime} \;\land\; \tau^{(t)}_{i, h-|w_i|+1\,:\,h}=w_i \;\land\; a^{(t)}_{i,h}=a_i \}}{\sum_{t=1}^{T} \mathbf{1}\{ \tau^{(t)}_{i, h-|w_i|+1\,:\,h}=w_i \land a_{i,h}^{(t)}=a_i \}}, \\[14pt]
\quad \text{ if } \sum_{t=1}^{T} \mathbf{1}\{ \tau^{(t)}_{h-|w_i|+1\,:\,h}=w_i \land a_{i,h}^{(t)}=a_i \} \ge 1 \text{ and } |w_i^{\prime}| = \min(m, |w_i|+1), \\[16pt]
0, \; \text{otherwise}.
\end{cases}
\end{align}

\paragraph{Rewards.} Similarly, at each player $i \in \mathcal{N}$, at step $h \in [H]$, and for all $w_i \in \mathcal{H}_i^{\leq m}$ and $a_i \in \mathcal{A}_i$, set
\begin{align}
\hat{r}_{i,h}^{m,\pi^{(k)}_{-i}}(w_i,a_i) \coloneqq \label{eq:sample-r}
\begin{cases}
\frac{\sum_{t=1}^{T} r_{i,h}^{(t)} \cdot \mathbf{1}\{ \tau^{(t)}_{i, h-|w_i|+1\,:\,h}=w_i \;\land\; a^{(t)}_{i,h}=a_i \}}
{\sum_{t=1}^{T} \mathbf{1}\{ \tau^{(t)}_{h-|w_i|+1\,:\,h}=w_i \;\land\; a^{(t)}_{i,h}=a_i \}},\\[14pt]
\qquad\qquad\qquad \text{ if } \sum_{t=1}^{T} \mathbf{1}\{ \tau^{(t)}_{i, h-|w_i|+1\,:\,h}=w_i \;\land\; a^{(t)}_{i,h}=a_i \} \ge 1, \\[12pt]
0, \; \text{otherwise}.
\end{cases}
\end{align}

\section{Proofs for Section~\ref{sec:results}}
\label{app:proof-results}

In this section, we provide the proof our main result, Theorem~\ref{thm:main}, on the convergence of Algorithm~\ref{alg:main} to an $\epsilon$-approximate Nash equilibrium.

\subsection{Model estimation guarantees}

We begin by deriving confidence bounds for the transition and reward estimates computed at each iteration of Algorithm~\ref{alg:main}. Note that these differ from the Markov game case, since we additionally need to account for the bias introduced by the superstate Markov game approximation.

\estproblemma*
\begin{proof}
We divide the proof into two parts for proving each of the two bounds.
\begin{itemize}
\item \textbf{Bound for estimation of transitions:} For some $\pi_i \in \Pi_i^H$, let $\widetilde{\mathbb{P}}^{\pi_i}_{i,h}(\cdot \mid w_i,a_i) \in \Delta(\mathcal{H}_i^{\leq m})$ denote the distribution over $m$-step windows observed at step $h$ marginalized over the full history, i.e.,
\begin{align*}
&\widetilde{\mathbb{P}}^{\pi_i}_{i,h}(w_i^{\prime} \mid w_i,a_i) \coloneqq \sum_{\tau_i=(o_{i,1},a_{i,1},\dots,a_{i,h},o_{i,h+1}) \in \mathcal{H}_i^{h+1}} P_{\pi_i} \left( \mathcal{T}_{i,h}=\tau_i \right) \\
&\hspace{11em}\cdot\mathbf{1} \left\{ \tau_{i,h-m+2 : h+1}=w_i^{\prime} \land \tau_{i,h-m+1 : h}=w_i \land a_{i,h+1}=a_i \right\}
\end{align*}
where $P_{\pi_i} \left( \mathcal{T}_{i,h}=\tau_i \right)$ is the probability that the random history $\mathcal{T}_{i,h}$, induced by the POMG dynamics and policy $\pi_i$, is equal to $\tau_i$ at step $h$ for player $i$.

Note that
\begin{align*}
\mathbb{E}_{\pi_i,s_{i,1} \sim \mu_i} \left[ \hat{\mathbb{P}}_{i,h}^{m,\pi_i}(w_i^{\prime} \mid w_i,a_i) \right] = \widetilde{\mathbb{P}}^{\pi^{\zeta}_i}_{i,h}(w_i^{\prime} \mid w_i,a_i),
\end{align*}
where $\pi^{\zeta}_i$ is the $\zeta$-greedy policy from which we sample in Line~\ref{ln:eps-greedy} of Algorithm~\ref{alg:main}. Being a ratio of two empirical counts over the $T$ sampled trajectories, $\hat{\mathbb{P}}_{i,h}^{m,\pi_i}(w_i^{\prime} \mid w_i,a_i)$ is itself not a sum of independent terms. Lemma~4 of~\cite{zhang2023markov} does not require this: it applies to empirical transition estimates of exactly this form, and is obtained by first rearranging the ratio into a quantity to which Hoeffding's inequality applies. Invoking it, for any $0 < \epsilon \leq 1$, we have
\begin{align*}
P \left( \left\vert \widetilde{\mathbb{P}}^{\pi^{\zeta}_i}_{i,h}(w_i^{\prime} \mid w_i,a_i)-\hat{\mathbb{P}}_{i,h}^{m,\pi_i}(w_i^{\prime} \mid w_i,a_i) \right\vert \geq \epsilon \right)
\leq 4 \exp \left( -\frac{\epsilon^{2} \zeta^{2m} \beta^{2m} T}{32 A_i^{2m}} \right).
\end{align*}
where the lower bound $c>0$ on state visitation probabilities required by this Lemma in our case is given by $c=\left( (\beta \zeta) / A_i \right)^m$. This is due to Assumption~\ref{ass:uni-obs} and the fact that we sample trajectories in a $\zeta$-greedy manner. At any step $h \in [H]$, any $m$-step window is visited with probability at least~$\left( (\beta \zeta) / A_i \right)^m$.

Moreover, by Lemma~\ref{lem:aux-p-r-err}, we have the deterministic bound
\begin{align*}
\left\vert \widetilde{\mathbb{P}}^{\pi_i^{\zeta}}_{i,h}(w_i^{\prime} \mid w_i,a_i) - \mathbb{P}^m_{i,h}(w_i^{\prime} \mid w_i,a_i) \right\vert \leq 2(1-\rho)^{m-1}
\end{align*}
which together with the above concludes the proof of the first bound.

\item \textbf{Bound for reward estimation:} Similar to above, we define
\begin{align*}
&\widetilde{r}_{i,h}^{\pi}(w_i,a_i) \coloneqq \sum_{\tau_i=(o_{i,1},a_{i,1},\dots,a_{i,h-1},o_{i,h}) \in \mathcal{H}^h_i} P_\pi \left( \mathcal{T}_{i,h}=\tau_i \right) \\
&\qquad\qquad\qquad\qquad\qquad \cdot \mathbf{1} \left\{ \tau_{i,h-m+1 : h}=w_i \land a_{i,h+1}=a_i \right\}\,r_{i,h}^{H,\pi_{-i}}(\tau_i,a_i),
\end{align*}
such that we have
\begin{align*}
\mathbb{E}_{\pi,s_1 \sim \mu} \left[ \hat{r}_{i,h}^{m,\pi_{-i}}(w_i,a_i) \right] = \widetilde{r}_{i,h}^{\pi^{\zeta}}(w_i,a_i),
\end{align*}
Similar as for transition probabilities, we obtain from Lemma~4 of \cite{zhang2023markov} the bound
\begin{align*}
P \left( \Big\vert \widetilde{r}_{i,h}^{\pi^{\zeta}}(w_i,a_i)-\hat{r}_{i,h}^{m,\pi_{-i}}(w_i,a_i) \Big\vert \geq \epsilon \right)
\leq 4 \exp \left( -\frac{\epsilon^{2} \zeta^{2m} \beta^{2m} T}{32 A_i^{2m}} \right).
\end{align*}
Moreover, we decompose the bias as follows
\begin{align*}
&\left\vert \widetilde{r}_{i,h}^{\pi^{\zeta}}(w_i,a_i) - r_{i,h}^{m,\pi_{-i}}(w_i,a_i) \right\vert \\
&\qquad\leq \underbrace{\left\vert \widetilde{r}_{i,h}^{\pi^{\zeta}}(w_i,a_i) - \widetilde{r}_{i,h}^{\pi}(w_i,a_i) \right\vert}_{(a)}
+ \underbrace{\left\vert \widetilde{r}_{i,h}^{\pi}(w_i,a_i) - r_{i,h}^{m,\pi_{-i}}(w_i,a_i) \right\vert}_{(b)}.
\end{align*}
Using Lemma~\ref{lem:r-error} below, for the second term, we obtain
\begin{align*}
(b) &\leq \left\vert \widetilde{r}_{i,h}^{\pi}(w_i,a_i) - r_{i,h}^{m,\pi_{-i}}(w_i,a_i) \right\vert \\
&\leq \sum_{\tau_i=(o_{i,1},a_{i,1},\dots,a_{i,h-1},o_{i,h}) \in \mathcal{H}^h_i} P \left( \mathcal{T}_{i,h}=\tau_i \right)
\mathbf{1} \left\{ \tau_{i,h-m+1 : h}=w_i \land a_{i,h+1}=a_i \right\} \\
&\hspace{20em} \cdot \left\vert r_{i,h}^{H,\pi_{-i}}(\tau_i,a_i) - r_{i,h}^{m,\pi_{-i}}(w_i,a_i) \right\vert \\
&\lesssim N^2H(1-\rho)^{m-1},
\end{align*}
and the first term is bounded by Lemma~\ref{lem:r-greedy-error},
\begin{align*}
(a) \leq 4 H (N-1)\zeta.
\end{align*}
Together with our probabilistic bound on $\vert \widetilde{r}_{i,h}^{\pi}(w_i,a_i)-\hat{r}_{i,h}^{m,\pi_{-i}}(w_i,a_i) \vert$, this concludes the proof.
\end{itemize}
\end{proof}

\begin{lemma}
\label{lem:r-greedy-error}
Let $\pi \in \Pi^m$, and for $\zeta > 0$ define $\pi^{\zeta} \in \Pi^m$ such that
\begin{align*}
\pi_i^{\zeta}(a_i \mid w_i) \coloneqq \frac{\zeta}{A_i} + (1-\zeta)\pi_i(a_i \mid w_i).
\end{align*}
Let also
\begin{align*}
r_{i,h}^{m,\pi}(w_i,a_i) \coloneqq \mathbb{E}_{w_{-i} \sim d^{m,\pi}_{-i,h}, \, a_{-i} \sim \pi_{-i}(\cdot \mid w_{-i})} \big[ r_{i,h}^m(w_i,w_{-i},a_i,a_{-i}) \big].
\end{align*}
Assume $r_{i,h}^m \in [0,1]$. Then, for any $w_i \in \mathcal{H}_i^{\leq m}$ and $a_i \in \mathcal{A}_i$,
\begin{align*}
\big| r_{i,h}^{m,\pi}(w_i,a_i) - r_{i,h}^{m,\pi^{\zeta}}(w_i,a_i) \big|
\le 4 H (N-1)\zeta.
\end{align*}
\end{lemma}

\begin{proof}
Fix $w_i,a_i$. We write
\begin{align*}
&\big| r_{i,h}^{m,\pi}(w_i,a_i) - r_{i,h}^{m,\pi^{\zeta}}(w_i,a_i) \big| \\[6pt]
&= \Bigg| 
\mathbb{E}_{\substack{w_{-i} \sim d^{m,\pi}_{-i,h} \\ a_{-i} \sim \pi_{-i}(\cdot \mid w_{-i})}} \big[ r_{i,h}^m(w_i,w_{-i},a_i,a_{-i}) \big]
-
\mathbb{E}_{\substack{w_{-i} \sim d^{m,\pi^\zeta}_{-i,h} \\ a_{-i} \sim \pi^\zeta_{-i}(\cdot \mid w_{-i})}} \big[ r_{i,h}^m(w_i,w_{-i},a_i,a_{-i}) \big]
\Bigg| \\[6pt]
&\le 
\underbrace{
\Bigg|
\mathbb{E}_{\substack{w_{-i} \sim d^{m,\pi}_{-i,h} \\ a_{-i} \sim \pi_{-i}(\cdot \mid w_{-i})}} \big[ r_{i,h}^m(w_i,w_{-i},a_i,a_{-i}) \big]
-
\mathbb{E}_{\substack{w_{-i} \sim d^{m,\pi}_{-i,h} \\ a_{-i} \sim \pi^\zeta_{-i}(\cdot \mid w_{-i})}} \big[ r_{i,h}^m(w_i,w_{-i},a_i,a_{-i}) \big]
\Bigg|
}_{(a)} \\
&\quad +
\underbrace{
\Bigg|
\mathbb{E}_{\substack{w_{-i} \sim d^{m,\pi}_{-i,h} \\ a_{-i} \sim \pi^\zeta_{-i}(\cdot \mid w_{-i})}} \big[ r_{i,h}^m(w_i,w_{-i},a_i,a_{-i}) \big]
-
\mathbb{E}_{\substack{w_{-i} \sim d^{m,\pi^\zeta}_{-i,h} \\ a_{-i} \sim \pi^\zeta_{-i}(\cdot \mid w_{-i})}} \big[ r_{i,h}^m(w_i,w_{-i},a_i,a_{-i}) \big]
\Bigg|
}_{(b)}.
\end{align*}
For term $(a)$, since $r_{i,h}^m \in [0,1]$, for any fixed $w_{-i}$,
\begin{align*}
&\Big|
\mathbb{E}_{a_{-i} \sim \pi_{-i}(\cdot \mid w_{-i})} [ r_{i,h}^m(w_i,w_{-i},a_i,a_{-i}) ]
-
\mathbb{E}_{a_{-i} \sim \pi^\zeta_{-i}(\cdot \mid w_{-i})} [ r_{i,h}^m(w_i,w_{-i},a_i,a_{-i}) ]
\Big| \\[4pt]
&\qquad\le 2\| \pi_{-i}(\cdot \mid w_{-i}) - \pi^\zeta_{-i}(\cdot \mid w_{-i}) \|_{TV} \\[2pt]
&\qquad\le 2(N-1)\zeta,
\end{align*}
and hence $(a) \le 2 (N-1)\zeta$.

For the second term, we have $(b) \le \| d^{m,\pi}_{-i,h} - d^{m,\pi^\zeta}_{-i,h} \|_{TV}$. By decoupledness of dynamics, $d^{m,\pi}_{-i,h}$ factorizes across players, and hence by induction on $h \in [H]$ similar to the proof of Lemma~\ref{lem:value-approx}, we obtain a bound of
\begin{align*}
(b) \le 2\sum_{j \in \mathcal{N} \setminus \{i\}} \| d^{m,\pi}_{j,h} - d^{m,\pi^\zeta}_{j,h} \|_{TV} \le 2H(N-1)\zeta.
\end{align*}
Combining the bounds on $(a)$ and $(b)$ gives the claimed inequality.
\end{proof}

\begin{lemma}
\label{lem:r-error}
Suppose Assumption~\ref{ass:filter-stab} holds. Let $\pi \in \Pi^m$, and for any $i \in \mathcal{N}$, $h \in [H]$, $w_i \in \mathcal{H}_i^{\leq m}$, $\tau_i \in \mathcal{H}_i^{h}$, and $a_i \in \mathcal{A}_i$, recall that marginal rewards are defined as
\begin{align*}
r_{i,h}^{m,\pi_{-i}}(w_i,a_i) &\coloneqq \mathbb{E}_{w_{-i} \sim d^{m,\pi}_{-i,h}, a_{-i} \sim \pi_{-i}(\cdot \mid w_{-i})} \left[ r_{i,h}^m(w_i,w_{-i},a_i,a_{-i}) \right],\\
r_{i,h}^{H,\pi_{-i}}(\tau_i,a_i) &\coloneqq \mathbb{E}_{\tau_{-i} \sim d^{H,\pi}_{-i,h}, a_{-i} \sim \pi_{-i}(\cdot \mid \tau_{-i})} \left[ r_{i,h}^H(\tau_i,\tau_{-i},a_i,a_{-i}) \right].
\end{align*}
Then, for any $w_i \in \mathcal{H}_i^{\leq m}$ and $\tau_i \in \mathcal{H}_i^H$ with $\tau_{i,h-m+1:h}=w_i$, it holds that
\begin{align*}
\left\vert r_{i,h}^{m,\pi_{-i}}(w_i,a_i)-r_{i,h}^{H,\pi_{-i}}(\tau_i,a_i) \right\vert
\lesssim N^2H(1-\rho)^{m-1}.
\end{align*}
\end{lemma}
\begin{proof}
Let $h \in [H]$. We decompose the difference as follows,
\begin{align*}
&\left\vert r_{i,h}^{m,\pi_{-i}}(w_i,a_i)-r_{i,h}^{H,\pi_{-i}}(\tau_i,a_i) \right\vert \\
&\quad= \Bigg\lvert \sum_{w_{-i} \in \mathcal{H}_{-i}^{\leq m},a_{-i} \in \mathcal{A}_{-i}} d^{m,\pi}_{-i,h}(w_{-i}) \cdot \pi_{-i}(a_{-i} \mid w_{-i}) \cdot r_{i,h}^m(w_i,w_{-i},a_i,a_{-i}) \\
&\qquad\qquad- \sum_{\tau_{-i} \in \mathcal{H}_{-i}^H,a_{-i} \in \mathcal{A}_{-i}} d^{H,\pi}_{-i,h}(\tau_{-i}) \cdot \pi_{-i}(a_{-i} \mid \tau_{-i}) \cdot r_{i,h}^H(\tau_i,\tau_{-i},a_i,a_{-i}) \Bigg\rvert \\
&\quad \overset{(a)}{\leq} \Bigg\lvert \sum_{w_{-i} \in \mathcal{H}_{-i}^{\leq m},a_{-i} \in \mathcal{A}_{-i}} d^{m,\pi}_{-i,h}(w_{-i}) \cdot \pi_{-i}(a_{-i} \mid w_{-i}) \cdot r_{i,h}^m(w_i,w_{-i},a_i,a_{-i}) \\
&\qquad\qquad- \sum_{\tau_{-i} \in \mathcal{H}_{-i}^H,a_{-i} \in \mathcal{A}_{-i}} d^{H,\pi}_{-i,h}(\tau_{-i}) \cdot \pi_{-i}(a_{-i} \mid \tau_{-i}) \cdot r_{i,h}^m(\tau_{i,h-m+1:h},\tau_{-i,h-m+1:h},a_i,a_{-i}) \Bigg\rvert \\
&\qquad\qquad+ 2N(1-\rho)^{m-1} \\
&\quad \overset{(b)}{\leq} \sum_{w_{-i} \in \mathcal{H}_{-i}^{\leq m}} \left\vert d^{m,\pi}_{-i,h}(w_{-i}) - \sum_{\tau_{-i,1:h-m} \in \mathcal{H}_{-i}^H} d^{H,\pi}_{-i,h}(\tau_{-i,1:h-m} \circ w_{-i}) \right\vert +2N(1-\rho)^{m-1} \\
&\quad \overset{(c)}{\leq} \sum_{j \in \mathcal{N} \setminus \{i\}} \sum_{w_j \in \mathcal{H}_j^{\leq m}} \left\vert d^{m,\pi}_{j,h}(w_j) - \sum_{\tau_{j,1:h-m} \in \mathcal{H}_j^H} d^{H,\pi}_{j,h}(\tau_{j,1:h-m} \circ w_j) \right\vert +2N(1-\rho)^{m-1} \\
&\quad \overset{(d)}{\leq} N \cdot HN(1-\rho)^{m-1} + 2N(1-\rho)^{m-1},
\end{align*}
where
\begin{itemize}
\item (a) uses the reward approximation bound from Lemma~\ref{lem:aux-p-r-err},
\item (b) is by boundedness of rewards and rearraging the summations,
\item (c) holds since by decoupledness of state transitions, we can factor state visitation distributions over players, and
\item (d) follows from bounding the error in state visitation distributions induced by $\mathbb{P}^m_i$ vs.\ $\mathbb{P}^H_i$ inductively over the $H$ steps (similar to the proofs in Lemma~\ref{lem:value-approx}), and using the transition probability approximation bound from Lemma~\ref{lem:aux-p-r-err}.
\end{itemize}
\end{proof}

Next, we provide an approximation bound for the $Q$ function obtained from backward dynamic programming in terms of reward and transition probability estimation errors.
\begin{lemma}
\label{lem:q-approx}
Suppose the bounds on transition and reward estimates from Lemma~\ref{lem:est-prob} hold, namely, we have
\begin{align*}
\left\vert \mathbb{P}^m_{i,h}(w_i^{\prime} \mid w_i,a_i)-\hat{\mathbb{P}}_{i,h}^{m,\pi_i}(w_i^{\prime} \mid w_i,a_i) \right\vert &\leq \epsilon + (1-\rho)^{m-1} \eqqcolon \epsilon_r,\\
\left\vert r_{i,h}^{m,\pi_{-i}}(w_i,a_i)-\hat{r}_{i,h}^{m,\pi_{-i}}(w_i,a_i) \right\vert &\leq \epsilon + 4NH\zeta + N^2 H (1-\rho)^{m-1} \eqqcolon \epsilon_\mathbb{P}.
\end{align*}
Then, for all $i \in \mathcal{N}$, $h \in [H]$, $k \in [K]$, $w_i \in \mathcal{H}_i^{\leq m}$, and $a_i \in \mathcal{A}_i$, it holds that
\begin{align*}
\left\vert \hat{Q}^{m,\pi}_{i,h}(w_i,a_i) - \overline{Q}_{i,h}^{m,\pi}(w_i,a_i) \right\vert \,\lesssim\,
H^3 N \cdot |\mathcal{A}_i| \cdot |\mathcal{O}_i| \cdot \left( \epsilon N + (1-\rho)^{m-1} \right).
\end{align*}
\end{lemma}
\begin{proof}
The proof proceeds by showing a stronger inequality by backward induction on $h$, namely that
\begin{align}
\label{eqn:stronger-ineq}
\left\vert \hat{Q}^{m,\pi}_{i,h}(w_i,a_i) - \overline{Q}_{i,h}^{m,\pi}(w_i,a_i) \right\vert
\leq \epsilon_r(H+1-h) + \epsilon_\mathbb{P} H(H+1-h) \cdot |\mathcal{A}_i| \cdot |\mathcal{O}_i|.
\end{align}
For $H+1$, (\ref{eqn:stronger-ineq}) trivially holds. Suppose (\ref{eqn:stronger-ineq}) holds for some $h \in [H]$. We can write the error as
\begin{align*}
\left\vert \left( \hat{Q}^{m,\pi}_{i,h} - \overline{Q}_{i,h}^{m,\pi} \right)(w_i,a_i) \right\vert
&\leq \left\vert \left( \hat{r}_{i,h}^{m,\pi_{-i}}-r_{i,h}^{m,\pi_{-i}} \right)(w_i,a_i) \right\vert \\
&\qquad + \sum_{w_i^{\prime},a_i^{\prime}} \mathbb{P}_{i,h}(w_i^{\prime} \mid w_i,a_i) \cdot \pi_{i,h+1}(a_i^{\prime} \mid w_i^{\prime}) \cdot \left( \hat{Q}^{m,\pi}_{i,h+1} - \overline{Q}_{i,h+1}^{m,\pi} \right)(w_i^{\prime},a_i^{\prime}) \\
&\qquad + \sum_{w_i^{\prime},a_i^{\prime}} \left( \hat{\mathbb{P}}_{i,h}-\mathbb{P}_{i,h} \right)(w_i^{\prime} \mid w_i,a_i) \cdot \pi_{i,h+1}(a_i^{\prime} \mid w_i^{\prime}) \cdot \hat{Q}^{m,\pi}_{i,h+1}(w_i^{\prime},a_i^{\prime}) \\
&\overset{(a)}{\leq} \epsilon_r + \left\lVert \hat{Q}^{m,\pi}_{i,h+1} - \overline{Q}_{i,h+1}^{m,\pi} \right\rVert_\infty + \epsilon_\mathbb{P} H \cdot |\mathcal{A}_i| \cdot |\mathcal{O}_i| \\
&\overset{(b)}{\leq} \epsilon_r + \left( \epsilon_r(H-h)+\epsilon_\mathbb{P}H(H-h) \cdot |\mathcal{A}_i| \cdot |\mathcal{O}_i| \right) + \epsilon_\mathbb{P} H  \cdot |\mathcal{A}_i| \cdot |\mathcal{O}_i| \\
&\leq \epsilon_r(H+1-h)+\epsilon_\mathbb{P}H(H+1-h) \cdot |\mathcal{A}_i| \cdot |\mathcal{O}_i|
\end{align*}
For (a) we use the fact that for most $m$-step windows, the transition probability error is $0$. Namely, for all $w_i^{\prime} \in \mathcal{H}_i^{\leq m}$ for which the first $m-1$ actions and observation do not coincide with the last $m-1$ actions and observations of $w_i$, we have $\hat{\mathbb{P}}_{i,h}(w_i^{\prime} \mid w_i,a_i) = \mathbb{P}_{i,h}(w_i^{\prime} \mid w_i,a_i) = 0$. This observation is crucial, as it saves us an factor exponential in $m$ one would obtain from summing over the entire space $\mathcal{H}_i^{\leq m}$. In (b) we apply the induction hypothesis. The final bound follows after plugging in~$\epsilon_r$ and~$\epsilon_\mathbb{P}$.
\end{proof}

Based on the marginalized Q-function, we define the advantage functions in $\mathcal{G}^m$ and in the Markov game underlying $\mathcal{P}$, respectively. For any $i \in \mathcal{N}$, $h \in [H]$, $\pi \in \Pi^m$, $w_i \in \mathcal{H}_i^{\leq m}$, $s_i \in \mathcal{S}_i$, and $a_i \in \mathcal{A}_i$, let
\begin{align*}
A^{m,\pi}_{i,h}(w_i,a_i) &\coloneqq \overline{Q}^{m,\pi}_{i,h}(w_i,a_i) - \sum_{a_i^{\prime} \in \mathcal{A}_i} \pi_{i,h}(a_i^{\prime} \mid w_i) \overline{Q}^{m,\pi}_{i,h}(w_i,a^{\prime}_i), \\
A_{i,h}^\pi(s_i,a_i) &\coloneqq \overline{Q}^\pi_{i,h}(s_i,a_i) - \sum_{a_i^{\prime} \in \mathcal{A}_i} \pi_{i,h}(a_i^{\prime} \mid s_i) \overline{Q}^\pi_{i,h}(s_i,a^{\prime}_i).
\end{align*}

Note for later use that the gain $\max_{a_i} A^{m,\pi}_{i,h}(w_i,a_i)$ is nonnegative for every $\pi \in \Pi^m$, $i \in \mathcal{N}$, $h \in [H]$, and $w_i \in \mathcal{H}_i^{\leq m}$, since $\sum_{a_i} \pi_{i,h}(a_i \mid w_i) A^{m,\pi}_{i,h}(w_i,a_i) = 0$ by definition.

The following lemma is the analogue of the gradient domination bound of Lemma~8 of~\cite{zhang2023markov} for the superstate game $\mathcal{G}^m$, transferred to the POMG $\mathcal{P}$ via the finite-window approximation of Lemma~\ref{lem:value-approx}. It is applied in the proof of Lemma~\ref{lem:soft-policy-iter} below (see step~(d)).
\begin{lemma}[Approximate gradient domination]
\label{lem:grad-dom}
Let Assumption~\ref{ass:filter-stab} hold, let $\pi \in \Pi^m$ and $i \in \mathcal{N}$, and let $\pi_i^{\star} \in \arg\max_{\pi_i^{\prime} \in \Pi^m_i} V_i(\pi_i^{\prime},\pi_{-i})$. Then
\begin{align}
\label{eqn:grad-dom}
\max_{\pi_i^{\prime} \in \Pi^m_i} V_i(\pi_i^{\prime},\pi_{-i}) - V_i(\pi)
\;\leq\; \sum_{h=1}^H \sum_{w_i \in \mathcal{H}_i^{\leq m}} d^{m,\pi_i^{\star}}_{i,h}(w_i)\, \max_{a_i} A^{m,\pi}_{i,h}(w_i,a_i) + 2\epsilon_\rho^m.
\end{align}
\end{lemma}
\begin{proof}
We first establish the corresponding bound within the superstate game $\mathcal{G}^m$. Since state transitions are decoupled, neither the kernel $\mathbb{P}^m_i$ nor the marginalized reward $r^{m,\pi_{-i}}_{i,h}$ of the MDP $\mathcal{M}^{\pi_{-i}}$ depends on player $i$'s policy, so that $V^m_i(\pi_i^{\prime},\pi_{-i})$ and $V^m_i(\pi)$ are the values of $\pi_i^{\prime}$ and $\pi_i$ in the same MDP $\mathcal{M}^{\pi_{-i}}$, whose $Q$-function under $\pi$ is $\overline{Q}^{m,\pi}_{i,h}$. Applying the performance difference lemma in $\mathcal{M}^{\pi_{-i}}$ therefore gives, for any $\pi_i^{\prime} \in \Pi^m_i$,
\begin{align}
V^m_i(\pi_i^{\prime},\pi_{-i}) - V^m_i(\pi)
&= \sum_{h=1}^H \sum_{w_i \in \mathcal{H}_i^{\leq m}} d^{m,\pi_i^{\prime}}_{i,h}(w_i) \sum_{a_i \in \mathcal{A}_i} \pi^{\prime}_{i,h}(a_i \mid w_i)\, A^{m,\pi}_{i,h}(w_i,a_i) \nonumber \\
&\leq \sum_{h=1}^H \sum_{w_i \in \mathcal{H}_i^{\leq m}} d^{m,\pi_i^{\prime}}_{i,h}(w_i)\, \max_{a_i} A^{m,\pi}_{i,h}(w_i,a_i),
\label{eqn:grad-dom-superstate}
\end{align}
since $\pi^{\prime}_{i,h}(\cdot \mid w_i)$ is a probability distribution over $\mathcal{A}_i$, and where we used that $d^{m,(\pi_i^{\prime},\pi_{-i})}_{i,h} = d^{m,\pi_i^{\prime}}_{i,h}$, as $d^{m,\pi}_{i,h}$ depends on $\pi$ only through $\pi_i$ by decoupledness.

It remains to transfer this bound from $\mathcal{G}^m$ to the POMG $\mathcal{P}$. Both $(\pi_i^{\star},\pi_{-i})$ and $\pi$ lie in $\Pi^m$, so applying the first bound of Lemma~\ref{lem:value-approx} to each of them yields
\begin{align*}
V_i(\pi_i^{\star},\pi_{-i}) - V_i(\pi)
\;\leq\; V^m_i(\pi_i^{\star},\pi_{-i}) - V^m_i(\pi) + 2\epsilon_\rho^m.
\end{align*}
Bounding the right-hand side by~(\ref{eqn:grad-dom-superstate}) at $\pi_i^{\prime} = \pi_i^{\star}$, and using $V_i(\pi_i^{\star},\pi_{-i}) = \max_{\pi_i^{\prime} \in \Pi^m_i} V_i(\pi_i^{\prime},\pi_{-i})$ gives~(\ref{eqn:grad-dom}).
\end{proof}

\begin{lemma}
\label{lem:soft-policy-iter}
Suppose Assumptions~\ref{ass:potential},~\ref{ass:filter-stab}, and~\ref{ass:concentr} hold, and that for some $0 < \epsilon \leq 1$ and for all $i \in \mathcal{N}$, $h \in [H]$, and $k \in [K]$, our $Q$-function estimates satisfies
\begin{align*}
\left\vert \hat{Q}^{m,\pi^{(k)}}_{i,h}(w_i,a_i) - Q_{i,h}^{m,\pi^{(k)}}(w_i,a_i) \right\vert \,\lesssim\,
H^3 N \cdot |\mathcal{A}_i| \cdot |\mathcal{O}_i| \cdot \left( \epsilon N + (1-\rho)^{m-1} \right) \coloneqq \epsilon_Q.
\end{align*}
Then, choosing stepsize $\eta^{(k)} = \frac{1}{\sqrt{4N^2H^3 k}}$, there exists $k \in [K]$ such that for all $i \in \mathcal{N}$,
\begin{align}
\label{eqn:soft-policy-iter}
\max_{\pi_i^{\prime} \in \Pi_i^m}V(\pi_i^{\prime},\pi_{-i}^{(k)}) - V(\pi^{(k)})
\lesssim \kappa \left( \frac{\sqrt{N^2 H^3} \left(HN + \log(K)\right) }{\sqrt{K} }+ NH\epsilon_Q + N^2 H^{5/2} \epsilon_\rho^m \sqrt{K} \right).
\end{align}
\end{lemma}
\begin{proof}
Fix some $k \in [K]$ and define for each $i \in \mathcal{N}$ the joint policy
\begin{align*}
\widetilde\pi^i:= (\pi_1^{(k)},\dots,\pi_{i-1}^{(k)},\pi_i^{(k+1)},\dots,\pi_N^{(k+1)}) \in \Pi^m.
\end{align*}
Then, since by Proposition~\ref{prop:alpha-potential}, the superstate game $\mathcal{G}^m$ is a Markov $2\epsilon_\rho^m$-potential game, we can decompose
\begin{align}
\label{eqn:improvement}
&\Psi(\pi^{(k+1)}) - \Psi(\pi^{(k)}) \\
&\qquad= \sum_{i=1}^N \Psi(\widetilde \pi^i) - \Psi(\widetilde \pi^{i+1})\\
&\qquad\geq -2N\epsilon_\rho^m + \sum_{i=1}^N V^m_i(\widetilde \pi^i) - V^m_i(\widetilde \pi^{i+1})\\
&\qquad\overset{(a)}{=} - 2N\epsilon_\rho^m + \sum_{i=1}^N \sum_{h=1}^H \sum_{w_i,a_i} d^{m,\widetilde\pi^i}_{i,h}(w_i)\, \big(\pi^{(k+1)}_i(a_i \mid w_i) - \pi^{(k)}_i(a_i \mid w_i)\big) \,Q^{\widetilde\pi^{i+1}}_i(w_i,a_i) \\
&\qquad\overset{(b)}{\geq} \eta^{(k)} \sum_{i=1}^N \sum_{h=1}^H \sum_{w_i} d^{m,\widetilde\pi^i}_{i,h}(w_i)\, \max_{a_i} A^{m,\pi^{(k)}}_{i,h}(w_i,a_i) - 2\eta^{(k)}NH \epsilon_Q - 4N^2 H^3 (\eta^{(k)})^2 - 2N \epsilon_\rho^m \\
&\qquad\overset{(c)}{\geq} \frac{\eta^{(k)}}{\kappa^{(k)}} \sum_{i=1}^N \sum_{h=1}^H \sum_{w_i} d^{m,\pi_i^{\star,k}}_{i,h}(w_i) \max_{a_i} A^{m,\pi^{(k)}}_{i,h}(w_i,a_i) - 2\eta^{(k)}NH \epsilon_Q - 4N^2 H^3 (\eta^{(k)})^2 - 2N \epsilon_\rho^m \\
&\qquad\overset{(d)}{\geq} \frac{\eta^{(k)}}{\kappa^{(k)}} \sum_{i=1}^N \left( \max_{\pi_i^{\prime} \in \Pi_i^m}V(\pi_i^{\prime},\pi_{-i}^{(k)}) - V(\pi^{(k)}) \right) - 2\eta^{(k)}NH \epsilon_Q - 4N^2 H^3 (\eta^{(k)})^2 - 4NH \epsilon_\rho^m
\end{align}
where (a) is by the performance difference lemma (see Lemma~6 of~\cite{zhang2023markov}), and (b) follows similar as in \cite{zhang2023markov}.
Step (c) compares the superstate visitation of the current iterate with that of the best response. This is where the coefficient $\kappa^{(k)}$ of Assumption~\ref{ass:concentr} enters: Since the $i$-th component of $\widetilde\pi^i$ is $\pi^{(k+1)}_i$ and $d^{m,\pi}_{i,h}$ depends on $\pi$ only through $\pi_i$, we have $d^{m,\widetilde\pi^i}_{i,h} = d^{m,\pi^{(k+1)}_i}_{i,h}$, and the definition of $\kappa^{(k)}$ gives $d^{m,\pi^{\star,k}_i}_{i,h}(w_i) \leq \kappa^{(k)} d^{m,\widetilde\pi^i}_{i,h}(w_i)$ for every $w_i$. As the gains $\max_{a_i} A^{m,\pi^{(k)}}_{i,h}(w_i,a_i)$ are nonnegative, this yields, for every $i \in \mathcal{N}$ and $h \in [H]$,
\begin{align*}
\sum_{w_i \in \mathcal{H}_i^{\leq m}} d^{m,\widetilde\pi^i}_{i,h}(w_i)\, \max_{a_i} A^{m,\pi^{(k)}}_{i,h}(w_i,a_i)
\;\geq\; \frac{1}{\kappa^{(k)}} \sum_{w_i \in \mathcal{H}_i^{\leq m}} d^{m,\pi^{\star,k}_i}_{i,h}(w_i)\, \max_{a_i} A^{m,\pi^{(k)}}_{i,h}(w_i,a_i),
\end{align*}
and summing over $h \in [H]$ gives step (c).

Step (d) follows from Lemma~\ref{lem:grad-dom}, applied at $\pi = \pi^{(k)}$ with $\pi_i^{\star} = \pi_i^{\star,k}$ the best response of Assumption~\ref{ass:concentr}. Multiplying~(\ref{eqn:grad-dom}) by $\eta^{(k)} / \kappa^{(k)} \leq 1$, summing over $i \in \mathcal{N}$, and using $H \geq 1$ shows that the additional error is absorbed by replacing $-2N\epsilon_\rho^m$ with $-4NH\epsilon_\rho^m$ in the last line of~(\ref{eqn:improvement}).

The final bound~(\ref{eqn:soft-policy-iter}) follows from a telescoping argument applied to~(\ref{eqn:improvement}), the fact that $\Psi(\pi) \leq HN$ for all $\pi \in \Pi^m$, and our choice of stepsize. Concretely, writing $G^{(k)} \coloneqq \sum_{i \in \mathcal{N}} \big( \max_{\pi_i^{\prime} \in \Pi^m_i} V(\pi_i^{\prime},\pi^{(k)}_{-i}) - V(\pi^{(k)}) \big) \geq 0$ and summing~(\ref{eqn:improvement}) over $k \in [K]$ gives
\begin{align*}
\Big( \min_{k \in [K]} G^{(k)} \Big) \sum_{k=1}^K \frac{\eta^{(k)}}{\kappa^{(k)}}
\;\leq\; HN + 2NH\epsilon_Q \sum_{k=1}^K \eta^{(k)} + 4N^2H^3 \sum_{k=1}^K \big(\eta^{(k)}\big)^2 + 4NH\epsilon_\rho^m K.
\end{align*}
Dividing by $\sum_{k=1}^K \eta^{(k)}/\kappa^{(k)} \geq \kappa^{-1} \sum_{k=1}^K \eta^{(k)}$, which holds by Assumption~\ref{ass:concentr}, together with $\sum_{k=1}^K \eta^{(k)} \gtrsim \sqrt{K/(N^2H^3)}$, and using $\sum_{k=1}^K (\eta^{(k)})^2 \lesssim \log(K)/(N^2H^3)$ then yields~(\ref{eqn:soft-policy-iter}).
\end{proof}

We are now ready to prove the main theorem.

\begin{proof}[Proof of Theorem~\ref{thm:main}]
If the minimum in the requirement on $m$ is attained by $H$, then $m = H$ and, as noted in Section~\ref{sec:mg-repr}, $\mathcal{G}^H$ is equivalent to $\mathcal{P}$, so that the finite-window approximation is exact and every term involving $\epsilon^m_\rho$ below vanishes. We may therefore assume throughout that $m \geq c_1 \cdot \rho^{-1} \log \big( \frac{\kappa NHAO}{\beta \epsilon} \big)$.
By Lemma~\ref{lem:soft-policy-iter}, there exists $k \in [K]$ such that for all $i \in \mathcal{N}$,
\begin{align*}
\max_{\pi_i^{\prime} \in \Pi_i^m}V(\pi_i^{\prime},\pi_{-i}^{(k)}) - V(\pi^{(k)})
\lesssim \underbrace{\frac{\kappa\sqrt{N^2 H^3} \left(HN + \log(K)\right)}{\sqrt{K}}}_{(a)} + \underbrace{\kappa NH\epsilon_Q}_{(b)} + \underbrace{\kappa N^2 H^{5/2} \epsilon_\rho^m \sqrt{K}}_{(c)}.
\end{align*}
We next show that for appropriately chosen iteration number $K$, trajectory length $T$, and window size $m$, the terms (a), (b), and (c) are all upper bounded by $\epsilon / 3$.

With our choices of $\zeta = \epsilon$, $T \gtrsim \frac{\kappa^2 A^{2m+2}O^2 N^2 H^6}{\beta^{2m} \epsilon^{2m+2}} \log(1 / \delta)$ and $m \gtrsim \rho^{-1} \log \big( \frac{\kappa HAON}{\beta \epsilon} \big)$, applying Lemma~\ref{lem:est-prob} and Lemma~\ref{lem:q-approx}, yields that with probability at least $1-\delta$, we have $\epsilon_Q \lesssim \frac{\epsilon}{\kappa NH}$ and hence $(b) \lesssim \epsilon / 3$.

Moreover, by choosing $K \gtrsim \frac{\kappa^2 N^4 H^5 A^2}{\epsilon^2}$, we ensure that both $(a) \lesssim \epsilon / 3$, as well as $(c) \lesssim \epsilon / 3$. Note that, in contrast to a uniform lower bound on the superstate visitation, $\kappa$ does not depend on $m$, so that the requirement on $m$ imposed by $(c)$, namely $\epsilon_\rho^m \lesssim \epsilon^2/(\kappa^2 N^4H^{5}A)$, is met by the logarithmic choice of $m$ above for any $\rho > 0$.

The overall sample complexity of Algorithm~\ref{alg:main} is given by
\begin{align*}
TK = \mathcal{O} \left( \frac{\kappa^4 A^2 O^2 N^6 H^{11}}{\epsilon^2} \left( \frac{A}{\beta \epsilon} \right)^{2m} \log(1/\delta) \right).
\end{align*}
Moreover, we have
\begin{align*}
\left( \frac{A}{\beta \epsilon} \right)^{2m} = \exp\!\left( \frac{2}{\rho} \log\frac{A}{\beta \epsilon} \cdot \log\!\left(\frac{\kappa H A O N}{\beta\epsilon}\right) \right) = \left(\frac{\kappa H A O N}{\beta \epsilon}\right)^{\frac{2}{\rho}\log(A/(\beta \epsilon))},
\end{align*}
from which we can conclude the bound
\begin{align*}
TK = \mathcal{O} \left( \frac{\kappa^4 A^2 O^2 N^6 H^{11}}{\epsilon^2}
\left(\frac{\kappa H A O N}{\beta\epsilon}\right)^{\frac{2}{\rho}\log(A/(\beta \epsilon))} \log(1/\delta) \right) .
\end{align*}
\end{proof}

\newpage

\section{Numerical experiments}
\label{app:experiments}

In this section we report simulation results for Algorithm~\ref{alg:main} in two decoupled POMG environments related to the applications discussed in Appendix~\ref{app:decoupled}.\footnote{The code for all experiments is available at \url{https://github.com/philip-jordan/NE-POMG}.} The first is a dynamic transmission game over a shared channel, in which the action and observation spaces are deliberately small so that exact Nash gaps against full-history deviations can be computed. The second is a cooperative gridworld navigation task with a congestion penalty, in which exact Nash gaps are out of reach and we instead report the exact expected episode return. In both environments we run Algorithm~\ref{alg:main} for different window lengths $m$.

\subsection{Dynamic transmission game}
\label{app:exp-transmission}

\paragraph{Setting.} We consider a minimal instance of the wireless network example of Appendix~\ref{app:decoupled}: two users share a communication channel and each maintains a local packet queue, with rewards coupled through the congestion of the shared channel. For $N=2$ players over horizon $H=6$, the POMG is as follows.
\begin{itemize}[topsep=2pt, itemsep=1pt]
\item \emph{States.} $s_i \in \mathcal{S}_i = \{0,1\}$, indicating whether player $i$'s queue is empty or holds a packet. The initial distribution $\mu_i$ is uniform on $\mathcal{S}_i$.
\item \emph{Actions.} $a_i \in \mathcal{A}_i = \{ \mathsf{wait}, \mathsf{transmit} \}$.
\item \emph{Rewards.} Writing $\mathsf{T}$ for $\mathsf{transmit}$, we set
\begin{align}
\label{eqn:exp-reward}
r_{i,h}(s,a) \coloneqq c_0 + c_1 \mathbf{1}\{ s_i = 1,\, a_i = \mathsf{T} \} - c_2 \mathbf{1}\{ a_i = \mathsf{T} \} - c_3 \mathbf{1}\{ a_1 = a_2 = \mathsf{T} \}.
\end{align}
This models a benefit $c_1$ for transmitting a queued packet, a fixed transmission cost $c_2$, and an interference penalty $c_3$ incurred when both users transmit at once. We use $c_0=0.5$, $c_1=0.3$, $c_2=0.15$, and $c_3=0.35$, so that $r_{i,h} \in [0,1]$.
\item \emph{Transitions.} Player $i$'s queue evolves independently of the other players: an empty queue receives a packet with probability $\lambda$; a held packet is delivered successfully with probability $q$ if the player transmits, and otherwise expires with probability $p_{\mathrm{e}}$ while the player waits. We use $\lambda=0.125$, $q=0.8$, and $p_{\mathrm{e}}=0.06$, so that packets are long-lived and the local state stays correlated over several steps.
\item \emph{Observations.} Player $i$ observes its true queue state, $o_i = s_i$, with probability $p_{\mathrm{o}} = 0.1$, and an observation drawn uniformly from $\mathcal{O}_i = \{0,1\}$ otherwise.
\end{itemize}

\paragraph{Assumptions.} Transitions and observations are decoupled by construction, Assumption~\ref{ass:uni-obs} holds with $\beta = (1-p_{\mathrm{o}})/2 = 0.45$, and Assumption~\ref{ass:potential} is verified below. For filter stability (Assumption~\ref{ass:filter-stab}), we numerically maximize the one-step contraction ratio of the belief update over pairs of beliefs and all action-observation pairs, which is tractable since $|\mathcal{S}_i| = 2$. The maximum, $1-\rho \approx 0.996$, coincides with the product of the largest observation likelihood ratio $(1+p_{\mathrm{o}})/(1-p_{\mathrm{o}})$ and the largest Dobrushin coefficient of the transition kernel, $\max\{|1-\lambda-p_{\mathrm{e}}|,\, |\lambda-(1-q)|\}$. Hence Assumption~\ref{ass:filter-stab} holds with $\rho \approx 0.004$.

\paragraph{Potential structure.} Each player's reward consists of a term depending only on its own state and action, $c_1 \mathbf{1}\{ s_i = 1,\, a_i = \mathsf{T} \} - c_2 \mathbf{1}\{ a_i = \mathsf{T} \}$, and the interference cost $c_3 \mathbf{1}\{ a_1 = a_2 = \mathsf{T} \}$, which is shared by both players. Adding up the players' own terms and counting the shared cost once gives the potential
\begin{align*}
\phi_h(s,a) = c_1 \big( \mathbf{1}\{ s_1 = 1,\, a_1 = \mathsf{T} \} + \mathbf{1}\{ s_2 = 1,\, a_2 = \mathsf{T} \} \big) &- c_2 \big( \mathbf{1}\{ a_1 = \mathsf{T} \} + \mathbf{1}\{ a_2 = \mathsf{T} \} \big)\\
&- c_3 \mathbf{1}\{ a_1 = a_2 = \mathsf{T} \}.
\end{align*}
Indeed, writing $j$ for the player other than $i$, a direct computation from~(\ref{eqn:exp-reward}) gives
\begin{align*}
r_{i,h}(s,a) - \phi_h(s,a) = c_0 - c_1 \mathbf{1}\{ s_j = 1,\, a_j = \mathsf{T} \} + c_2 \mathbf{1}\{ a_j = \mathsf{T} \},
\end{align*}
which depends neither on $s_i$ nor on $a_i$. This is exactly the stage potential condition of Assumption~\ref{ass:potential}. Note that the game is not of identical interest, as players disagree on who should be the one to transmit.

\paragraph{Evaluation.} We report the exact Nash gap of the iterates $\pi^{(k)}$ against \emph{full-history} deviations,
\begin{align*}
\max_{i \in \mathcal{N}} \Big[ \max_{\pi_i^{\prime} \in \Pi_i^H} V_i(\pi_i^{\prime}, \pi_{-i}^{(k)}) - V_i(\pi^{(k)}) \Big],
\end{align*}
which is the quantity our learning objective in Section~\ref{sec:algorithm} controls. This is feasible here because, by decoupledness, fixing $\pi_{-i}$ leaves player $i$ facing a POMDP over its own two-state chain in which the other player enters only through the mean rewards $r_{i,h}^{H,\pi_{-i}}$, and the resulting best-response value can be computed by backward induction with exact belief propagation over player $i$'s history tree, which has $(|\mathcal{A}_i| \cdot |\mathcal{O}_i|)^H = 4^6 = 4096$ leaves.

\paragraph{Experiment.} We run Algorithm~\ref{alg:main} for window lengths $m \in \{1,\dots,6\}$ with $K=500$ iterations, $T=10^5$ episodes per iteration, exploration $\zeta = 0.02$, and stepsizes $\eta^{(k)} = \min\{1, \tfrac{1}{2}(k+1)^{-1/2}\}$, over $8$ independent runs. All runs start from the uniform policy and differ only in the random seed used for sampling episodes, i.e., in the realized transitions, observations, and actions.

\paragraph{Choice of parameters.} We set the parameters by hand, using a few trial runs. The stepsizes follow the $k^{-1/2}$ decay of Theorem~\ref{thm:main}. We chose $K$ large enough for all curves to flatten, and $T$ and $\zeta$ such that the statistical error and the exploration perturbation are small compared to the finite-window approximation error, so that different window lengths can be distinguished. These values are far smaller than those prescribed by Theorem~\ref{thm:main}, whose bounds are worst-case.

\paragraph{Results.} Figure~\ref{fig:nashgap} shows the Nash gap over the first $100$ iterations, by which point all curves have flattened; the gaps at $k=500$ agree with those at $k=100$ to within the plotted band. For every window length, the gap drops by one to two orders of magnitude within roughly $20$ iterations and then settles at an $m$-dependent floor which decreases monotonically in $m$. This is the behavior predicted by Proposition~\ref{prop:approx-ne} and Theorem~\ref{thm:main}: larger windows shrink the finite-window approximation error and therefore the accuracy at which independent learning saturates. The curves for $m=5$ and $m=6$ coincide, since with $H=6$ a window of five action-observation pairs already covers the entire history available to a player, so that no further increase of $m$ can help.

\begin{figure}[H]
\centering
\includegraphics[width=0.5\textwidth]{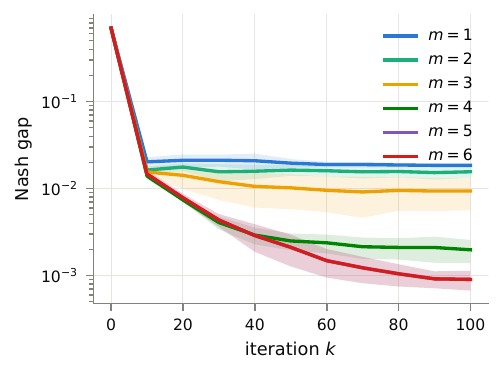}
\caption{Transmission game: exact Nash gap against full-history deviations over the iterations of Algorithm~\ref{alg:main}, for window lengths $m \in \{1,\dots,6\}$ ($N=2$, $H=6$, $T=10^5$). Curves show the mean over $8$ independent runs and shaded areas one standard deviation. The $m=5$ curve is hidden underneath the $m=6$ curve.}
\label{fig:nashgap}
\end{figure}

\subsection{Cooperative gridworld}
\label{app:exp-gridworld}

\paragraph{Setting.} Two agents, starting in cells \circled{$1$} and \circled{$2$}, navigate a $3 \times 3$ grid with two walls (shaded) towards a common goal cell \circled{$G$}, while congestion is penalized. This instantiates the multi-robot example of Appendix~\ref{app:decoupled}: transitions are decoupled, so agents may occupy the same cell, but sharing a cell is penalized. Since the only route between the two halves of the grid is the narrow passage in the center cell, the agents have to coordinate their crossings.

\begin{center}
\begin{tikzpicture}[scale=0.7, every node/.style={font=\small}]
  \foreach \r in {0,1,2}{
    \foreach \c in {0,1,2}{
      \draw (\c,2-\r) rectangle ++(1,1);
    }
  }
  \fill[black!18] (1,2) rectangle ++(1,1);
  \fill[black!18] (1,0) rectangle ++(1,1);
  \draw (1,2) rectangle ++(1,1);
  \draw (1,0) rectangle ++(1,1);
  \node[shape=circle, draw, inner sep=1.2pt] at (0.5,2.5) {$1$};
  \node[shape=circle, draw, inner sep=1.2pt] at (2.5,2.5) {$G$};
  \node[shape=circle, draw, inner sep=1.2pt] at (0.5,0.5) {$2$};
\end{tikzpicture}
\end{center}

Over horizon $H=10$, the POMG is as follows.
\begin{itemize}[topsep=2pt, itemsep=1pt]
\item \emph{States.} $s_i \in \mathcal{S}_i$ is player $i$'s cell, one of the $7$ non-wall cells. Player $1$ starts in the top-left cell and player $2$ in the bottom-left cell, so $\mu_i$ is a point mass. The goal cell $G$ is absorbing.
\item \emph{Actions.} $\mathcal{A}_i = \{ \mathsf{up}, \mathsf{down}, \mathsf{left}, \mathsf{right}, \mathsf{stay} \}$. A move into a wall or off the grid leaves the agent in place.
\item \emph{Rewards.} A player receives $+10$ at the step it first reaches $G$ and $0$ at every step thereafter. Before reaching $G$, it receives $-5$ if it shares its cell with the other player and $-1$ otherwise. For readability we state these rewards unnormalized; rescaling them into $[0,1]$ by a common affine map, as assumed in Section~\ref{sec:problem}, leaves the potential structure and the set of (approximate) Nash equilibria unchanged up to the corresponding rescaling of $\epsilon$.
\item \emph{Transitions.} Each player moves according to its own action, but with probability $\mathsf{slip} = 0.1$ the executed action is replaced by one drawn uniformly from $\mathcal{A}_i$.
\item \emph{Observations.} Player $i$ observes its true cell with probability $p_{\mathrm{o}} = 0.5$, and a uniformly random non-wall cell otherwise, so that $|\mathcal{O}_i| = 7$.
\end{itemize}

\paragraph{Assumptions.} Transitions and observations are decoupled by construction, Assumption~\ref{ass:uni-obs} holds with $\beta = (1-p_{\mathrm{o}})/7 = 1/14$, and Assumption~\ref{ass:potential} is verified below. Filter stability (Assumption~\ref{ass:filter-stab}), however, does not hold. The reason is the absorbing goal combined with informative observations: a belief concentrated on $G$ remains there forever, whereas observations of cells far from $G$ keep driving any other belief away from $G$. This experiment evaluates Algorithm~\ref{alg:main} beyond the scope of the approximation guarantees provided by our theory.

\paragraph{Potential structure.} Write $g_{i,h}$ for the part of player $i$'s reward that depends only on its own trajectory: the goal bonus $+10$ at the step it first reaches $G$, and the step cost $-1$ at every step before. The players are coupled only through the additional co-location penalty $-4 \cdot \mathbf{1}\{ s_1 = s_2 \neq G \}$, which both players incur, so that $r_{i,h} = g_{i,h} - 4 \cdot \mathbf{1}\{ s_1 = s_2 \neq G \}$. As in Appendix~\ref{app:exp-transmission}, Assumption~\ref{ass:potential} then holds with potential
\begin{align*}
\phi_h = g_{1,h} + g_{2,h} - 4 \cdot \mathbf{1}\{ s_1 = s_2 \neq G \},
\end{align*}
since $r_{i,h} - \phi_h = -g_{j,h}$ for $j \neq i$ does not depend on player $i$'s state or action. The game is not of identical interest: at the central passage each player prefers the other to yield.

\paragraph{Evaluation.} Here, computing the exact Nash gap is out of reach, since determining a full-history best response would require a dynamic program over $(|\mathcal{A}_i| \cdot |\mathcal{O}_i|)^H = 35^{10} \approx 2.8 \cdot 10^{15}$ histories. We therefore report the expected episode return of the iterate $\pi^{(k)}$ (without exploration), summed over both players, where the expectation is over the slip in the transitions, the observation noise, and the randomness of the policies.

\paragraph{Experiment.} We run Algorithm~\ref{alg:main} for window lengths $m \in \{1,2,3\}$ with $K=100$ iterations, $T=10^6$ episodes per iteration, exploration $\zeta = 0.1$, and stepsizes $\eta^{(k)} = \min\{1, \tfrac{1}{2}(k+1)^{-1/2}\}$, over $8$ independent runs, which, as in Appendix~\ref{app:exp-transmission}, differ only in the random seed. As before, the parameters were set by hand: the stepsizes are as in Appendix~\ref{app:exp-transmission}, $\zeta$, $K$ and $T$ were chosen after a few trial runs. All runs complete within about an hour on a consumer-grade CPU when the independent runs are executed in parallel.

\paragraph{Results.} Figure~\ref{fig:gridworld} shows the expected episode return over the iterations of Algorithm~\ref{alg:main}. In the left panel, the shortest window $m=1$ is insufficient for learning policies that reach the goal: the return improves over the initial uniform policy but stays deeply negative, around $-19$. For $m=2$ and $m=3$, the return instead climbs to a positive value of about $+5$, i.e., the players learn to reach the goal while avoiding congestion in the narrow central passage. The right panel zooms in on these two window lengths and shows that the longer window $m=3$ converges to a higher return, about $5.32$, than $m=2$, about $5.15$, and does so with a visibly smaller spread across runs.

\begin{figure}[H]
\centering
\begin{subfigure}{0.48\textwidth}
  \centering
  \includegraphics[width=\textwidth]{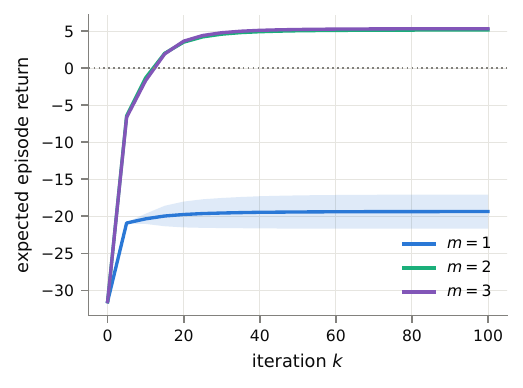}
  \caption{All window lengths $m \in \{1,2,3\}$.}
  \label{fig:gridworld-all}
\end{subfigure}
\hfill
\begin{subfigure}{0.48\textwidth}
  \centering
  \includegraphics[width=\textwidth]{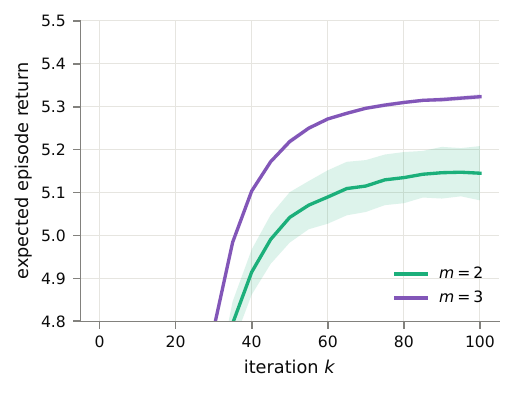}
  \caption{Zoom on $m \in \{2,3\}$.}
  \label{fig:gridworld-zoom}
\end{subfigure}
\caption{Gridworld: exact expected episode return, summed over both players, over the iterations of Algorithm~\ref{alg:main} ($N=2$, $H=10$, $T=10^6$). Curves show the mean over $8$ independent runs and shaded areas one standard deviation. Left: $m=1$ fails to reach the goal, while $m \geq 2$ does. Right: among the successful window lengths, the longer history $m=3$ attains the higher return.}
\label{fig:gridworld}
\end{figure}

\clearpage

\end{document}